\documentclass[showpacs,amsmath,a4paper,twocolumn, superscriptaddress,floatfix,nofootinbib]{revtex4-1}

\usepackage[sans]{dsfont} 
\usepackage{bbm}
\usepackage[mathscr]{euscript}

\usepackage{ucs}
\usepackage{graphicx}
\usepackage{url}

\usepackage{fullpage}

\usepackage[usenames,dvipsnames]{xcolor}

\usepackage{tcolorbox}
\tcbuselibrary{theorems} 

\usepackage{amsmath}
\usepackage{amssymb}
\usepackage{amsthm}
\usepackage{units}
\usepackage{framed}
\usepackage{mdframed}
\usepackage{thmtools, thm-restate} 

\usepackage{cancel} 

\usepackage[ colorlinks = true, 
             linkcolor = MidnightBlue,
             urlcolor  = MidnightBlue,
             citecolor = orange,
             anchorcolor = ForestGreen,
]{hyperref}

\makeatletter
\def\l@subsection#1#2{}
\def\l@subsubsection#1#2{}
\makeatother

\tcbset{colback=orange!5,colframe=orange!75!yellow, 
fonttitle= \bfseries, float=htb}

\newtcbtheorem[]{bigexample}{Example}{float*=htb, width=\textwidth}{ex}

\newtcolorbox[blend into=figures]{boxfigure}[2][]
{ float*=htb,width=\textwidth,lower separated=false, center upper, 
center title,title={#2},every float=\centering,#1}

\declaretheoremstyle[shaded={rulecolor=MidnightBlue,rulewidth=1pt, bgcolor={rgb}{1,1,1}}]{boxed}

\declaretheoremstyle[shaded={rulecolor=Thistle,rulewidth=1pt, bgcolor={rgb}{1,1,1}}]{secboxed}

\declaretheoremstyle[shaded={rulecolor=YellowOrange,rulewidth=1pt, bgcolor={rgb}{1,1,1}}]{terboxed}

\declaretheoremstyle[shaded={rulecolor=Green,rulewidth=1pt, bgcolor={rgb}{1,1,1}}]{tetraboxed}

\declaretheorem[style=boxed]{theorem}

\declaretheorem[within=section]{lemma}

\declaretheorem[sibling=lemma]{corollary}

\declaretheorem[sibling=lemma, style=tetraboxed]{proposition}
\declaretheorem[style=secboxed, sibling=lemma]{definition}

	\newcommand{\green}[1]{\textcolor{OliveGreen}{#1}}

	\newcommand{\flag}[1]{\green{ [#1]}}

    \newcommand{\ket}[1]{\vert  #1 \rangle}
    \newcommand{\bra}[1]{\langle #1 |}
    
	\newcommand{\proj}[2]{\ket{#1}\bra{#2}}
	\newcommand{\pure}[1]{\proj{#1}{#1}}

	\newcommand{\hilbert}{\mathcal{H}}

	\newcommand{\rank}{\operatorname{rank}}

	\newcommand{\id}{\mathbbm{1}}

	\newcommand{\ball}{\mathcal{B}}

	\newcommand{\tr}{\operatorname{Tr}  }
	
\newcommand*{\I}{\mathcal{I}}
\newcommand*{\E}{\mathcal{E}}

\newcommand{\cT}{\mathcal{T}}

\renewcommand{\P}{{\mathbb P}} 

\newcommand{\local}[1]{\widehat{#1}}

\newcommand{\Cost}[1]{\operatorname{Cost} \left( #1\right)_{\mathcal C} }
\newcommand{\Yield}[1]{\operatorname{Yield} \left( #1\right)_{\mathcal C} }
\newcommand{\Balance}[2]{\operatorname{Balance}(#1 \to  #2)_{\mathcal C}}

\newcommand{\BalanceII}[3]{\operatorname{Balance}(#1 \to  #2 | #3)_{\mathcal C}}

\newcommand{\val}{\operatorname{Val}}

	\newcommand{\hmin}{ H_{\min} }

	\newcommand*{\eps}{\varepsilon}

\begin{document}

\title{Currencies in resource theories}

\author{Lea \surname{Kr\"amer}}
\email{lkraemer@phys.ethz.ch}
\affiliation{Institute for Theoretical Physics, ETH Zurich, Switzerland}

\author{L\'idia \surname{del Rio}}

\email{lidia.delrio@bristol.ac.uk}
\affiliation{School of Physics, University of Bristol, United Kingdom}

\begin{abstract}
How may we quantify the value of physical resources, such as entangled quantum states, heat baths or lasers? 
Existing resource theories give us 
partial answers; however, these  rely on idealizations, like perfectly independent copies of states or exact knowledge of a quantum state.
Here we introduce the general tool of \emph{currencies} to quantify realistic descriptions of resources, applicable in experimental settings 
when we do not have perfect control over a physical system, when only the neighbourhood of a state or some of its properties are known, or when there is no obvious way to decompose a global space into subsystems. 
Currencies are a set of resources chosen to quantify all others --- like Bell pairs in LOCC or a lifted weight in thermodynamics. 
We show that from very weak assumptions on the theory we can already find useful currencies that give us necessary and sufficient conditions for resource conversion, and we build up more results as we impose further structure.
This work is an application of `Resource theories of knowledge'~\cite{DelRio2015}, generalizing axiomatic approaches to thermodynamic entropy~\cite{Lieb1999,Lieb2003,Lieb2013,Lieb2014}, work~\cite{Gallego2015}  and currencies made of local copies \cite{Coecke2014, Fritz2015}.

\end{abstract}

\maketitle

\tableofcontents
\newpage

\section*{Introduction}
\label{sec:introduction}
Every time you cut your losses, every time you traded marbles, you used resource theories. 
They look at the world and see games, sorting what is cheap and what is precious, the easy from the difficult. They are the foundation for every optimization problem and every impossibility statement. How so? Resource theories are  frameworks that allow agents to model the  world where they move in a subjective, operational way: 
they describe the systems available to the agents, and the constraints limiting their actions. Using only these simple ingredients, resource theories enable us to derive rules for resource transformations | what agents can do, and how they may act efficiently.

Resource theories can be used to model all kinds of constraints on the actions of an agent, ranging from physical laws to technical limitations or even rules of a game. 
Because of this, they have proven successful models in many different contexts, from physics to computer science and cryptography. 
For example, in entanglement theory the resource theory of local operations and classical communications (LOCC)~\cite{Bennet1996,Bennet1996a,Popescu1997,Horodecki2009}  has provided powerful tools to quantify entanglement. In quantum thermodynamics, the resource theories of noisy and thermal operations~\cite{Brandao2011a, Horodecki2003,Janzing2000,Horodecki2013a,Renes2014,Goold2015} have shed light on irreversibility and the second law, while  resource theories of asymmetry and reference frames~\cite{Vaccaro2003, Vaccaro2008, Bartlett2007,Gour2008,Marvian2013,Marvian2014,Marvian2015}, coherence~\cite{Baumgratz2013,Lostaglio2014b,Winter2015,Lostaglio2015,Korzekwa2015} and quantum control~\cite{Matera2015} explore the role   of other resources in quantum information theory.

Here we extend the range of applicability of resource theories, by showing that we do not need most of the assumptions of usual approaches in order to derive useful tools to quantify resources.

\subsection{Quantifying resources}
Resource theories 
start from a set of allowed transformations that an agent is able to implement, and derive
a structure that encodes which resources can be converted into which other resources under such transformations. 
Formally, a resource theory gives rise to a pre-order $\to$ on the set of resources, i.e.\ a reflexive and transitive relation. 
While this pre-order structure carries all the information needed to characterize  the resource theory, it is often difficult to compute relevant properties. 
For example, it can be hard to decide for any two given resources $A$ and $B$ whether or not  $A\to B$. 
Similarly, if it is not the case that $A \to B$, 
which additional resources would the agent have to supply in order to generate the resource $B$? 
Or if it is the case that $A\to B$, what kind of additional resources could be extracted along the way? 

This is where the concept of assigning \emph{value} to resources comes in. The idea is to quantify in a simple way how much can be done with a particular resource, and how difficult it is to transform one resource into another. 
This concept is hence central to resource theories: if we can determine a good characterization of value, then we can answer important questions in the theory. In equilibrium thermodynamics, for example, the free energy of a state quantifies the work that can be extracted from it, and is also a monotone under isothermal transformations. More generally, we need more than one monotone to characterize the structure of a resource theory.

\subsection{Currencies}
Looking at particular resource theories, we can find a common and powerful tool that is used to characterize the value of resources in an operational way. 
Namely, there is usually a special class of 
\emph{standard} or {reference} resources 
that 
one can measure or quantify more easily, and that can be used to determine the value of other resources.
What is more, these special resources are usually \emph{universal} in the sense that they allow an agent to 
generate any other resource from them. 
Such resources are then particularly useful to measure the value of other resources, since we can always convert back and forth between the resources of interest and the universal reference resources. 
We call such reference resources a \emph{currency}|they provide a natural operational interpretation of value.

In real life, for instance, \emph{money} constitutes a currency, since we can use it to trade against essentially any resource. 
For a more physical example, note that we can also identify a currency in two-party {\sc LOCC}: 
maximally entangled pairs of qubits (Bell pairs) 
can be used to 
teleport half of any quantum state~\cite{Bennet1996}, and so the value of resources like bipartite states can be defined in terms of how many Bell pairs are needed to create, or how many can be extracted from, a resource. These considerations yield the well-known concepts of entanglement of formation and distillable entanglement. 
Similarly, work in traditional thermodynamics 
can be considered a currency that allows to implement essentially any transformation if given enough of it. 
Finally, the concept of currencies  is also underlying the ideas in Lieb and Yngvason's approach to defining entropy in thermodynamics~\cite{Lieb1999,Lieb2003,Lieb2013,Lieb2014}, following earlier axiomatic work on thermodynamics~\cite{Caratheodory1909,Giles1964}; there, the set of equilibrium states forms a currency.

Now, if a resource theory has a currency, we can use it to find necessary and sufficient conditions for resource transformations. For instance, if the yield of ``selling'' the initial resource is larger than the cost of creating the final one, then the transformation is possible: we simply sell the initial resource and use some of the coins to buy the final one. 
In addition, the cost and yield of resources are also \emph{monotones} on the pre-order $\to$ of the theory. As they have to be decreasing along the pre-order, this yields additional necessary conditions for resource transformations. For a further discussion on monotones and their connection to currencies, see Appendix~\ref{appendix:monotones}.

\subsection{Previous approaches}
While currencies serve the same purpose in many resource theories, they are built on a number of different assumptions. 
For example, Bell pairs in LOCC rely on the concept of \emph{copies} of many resources, while work in (quantum) thermodynamics can be modelled as a pure state on an energy ladder, often satisfying a requirement of \emph{translational invariance}~\cite{Skrzypczyk2014,Aberg2013,Masanes2014,Alhambra2016}. 
In contrast, the currency or ``reference states'' in the axiomatic approach to thermodynamics of Refs.~\cite{Lieb1999,Lieb2003,Lieb2013,Lieb2014} satisfy special comparability and scalability properties.

But which assumptions are essential to guaranteeing necessary or sufficient conditions for resource transformations? Which assumptions do we need in order to characterize the cost of transformations? 
In this work, we will take a more abstract view on currencies and discuss more precisely which properties allow us to make which particular statements. 
Preceding this article, there has been 
some work towards abstract characterizations of resource theories and common tools like monotones and conversion rates therein~\cite{Brandao2015, Coecke2014, Fritz2015}. 
However, 
these works and the very concepts they seek to describe 
rely on idealizations such as 
exact knowledge of the state on the relevant systems
and 
complete absence of correlations between subsystems, which cannot always be guaranteed in practice.
In quantum theory, 
exact and perfectly uncorrelated states are the equivalent of spherical cows in a vacuum:  idealizations that ease the maths but cannot be implemented in the lab (for example, it is impossible to control the conditions of an experiment so well as to claim that a quantum system is in a given pure state with probability one).
Indeed, such abstractions may be of little use to experimentalists wanting to apply a physical theory to the resources at hand. To be relevant, resource theories must  realistically model the limitations of real agents, both in their descriptions of resources and in the constraints in action~\cite{Halpern2014b}. 

In Ref.~\cite{DelRio2015} we have developed a general framework for resource theories that allows us to overcome these issues, but we left the question open as to whether it is possible to recover the useful tools of traditional resource theories without all the handy simplifications.

\subsection{Contribution of this work}

In Section~\ref{sec:nontechnical} we provide an abstract characterization of currencies and identify minimal key assumptions needed to derive useful notions of value of resources and transformations. We show that we can make relevant statements without many of the properties that are taken for granted in other works.

In Section~\ref{sec:technical} we use the framework of Ref.~\cite{DelRio2015} to apply these tools to  general general settings where we may not have a known subsystem structure, a composition operation or a precise description of resources.
This paves the way to studying currencies that are themselves made up of more general resources, such as approximate Bell pairs or systems where slight correlations cannot be ruled out, and that lack properties like scalability as assumed
in~\cite{Lieb1999,Lieb2003,Lieb2013,Lieb2014}.
Then we present a concrete application of our tools to the resource theory of unital maps, where we quantify the cost, yield and balance for general resources and transformations. These quantities correspond to notions of \emph{entropy} for specifications, and yield direct results for the questions of work extraction and transformation balance for general resources. 

\section{Types of currencies}
\label{sec:nontechnical}
We split this discussion into three stages,  starting from a basic definition of currencies and adding in stronger conditions as we progress. 
We outline the properties and results that follow  from the assumptions made up to each stage. 
Technical definitions, statements and proofs can be found in Appendix~\ref{appendix:value}.

\subsection{Stage I: a universal standard}

At a basic level, currencies have two defining properties: \emph{order} 
and \emph{universality} for a given target. 
We show that such a basic notion of currencies already allows for a definition of \emph{cost} and \emph{yield} of all resources, from which we may derive necessary and  sufficient conditions for general resource transformations.

\subsubsection{Ordered set and value}
We have argued above that currencies are a collection of  `standard' or `reference' resources relative to which we can quantify other resources. 
Formally, we can define a currency as a subset $\mathcal{C}$ of resources equipped with a well-defined value function $\val:  \mathcal{C} \to \mathbb{R}_+ $  that satisfies
\begin{align*}
      C \to C' \iff& \val (C) \geq \val(C'), 
\end{align*}
for all currencies $C, C' \in \mathcal C$.\footnote{We consider functions to the positive real numbers |  we  assume the currency to have a finite least valuable resource, a \emph{bottom} (largely motivated by the framework of~\cite{DelRio2015}, where it corresponds to `knowing nothing'). 
Negative values can then be eliminated without loss of generality through an additive constant.}
In particular, the set $\mathcal{C}$ of currency resources must be totally ordered, up to equivalences (e.g.\ two different Bell states have the same value under LOCC, and in general $ C \rightleftharpoons C' \iff \val (C) = \val(C') $).
This condition guarantees that  a currency of high value is really strictly more powerful than one with a lower value (Proposition~\ref{prop:value_operational}).

While this may seem like a strong requirement at first, we observe that essentially all of the examples of currencies we find in resource theories satisfy this order requirement. Our order condition is also not unusual: it is in fact reminiscent of the comparability requirement in axiomatic approaches to thermodynamics~\cite{Giles1964,Lieb1999,Lieb2003,Lieb2013,Lieb2014}, where it is demanded for the set of reference states (`equilibrium states' or `entropy meter'). Finally, note that we can always start from the whole pre-order structure of the resource theory and build a currency by selecting a maximal subset of resources that is totally ordered (plus equivalences); any monotone function to $\mathbb R_+$ on this set is a valid value function.

\subsubsection{Universality}
We have mentioned that a currency is typically universal 
in the sense that any resource can be generated from the currency if given enough of it, and any resource can be ``sold'' for some currency resource (which could however be of zero value). While we usually find that this holds for all resources, we relax this condition slightly to hold for a target set $\mathcal{S}$ of resources. 
Formally, we can then write the property of \emph{universality} as follows:
for any 
resource $A$ in a target set $\mathcal{S}$,
there is a currency resource $C\in\mathcal{C}$ such that 
$ C \to A $
and a currency resource
$ C' \in \mathcal{C} $
such that
$ A \to C' $. 

In the axiomatic approach to thermodynamics of~\cite{Lieb2013}, this property corresponds to an extended comparability axiom for states out of thermodynamic equilibrium with equilibrium states.

\subsubsection{Cost and yield of resources}
Order and universality of a currency allow us to define the cost and yield of any target resource  $A\in\mathcal{S}$:
\begin{align*}
    \Cost{A} &= \inf_{C\in\mathcal{C}} (\val C: C \to A) \\
    \Yield{A} &= \sup_{C'\in\mathcal{C}} (\val C': A \to C'). 
\end{align*}
The cost and yield of resources are monotones --- and so we immediately obtain the following necessary conditions for resource transformations in the target $\mathcal{S}$ (Theorem~\ref{thm:currency_conditions}):
\begin{align*}
    A \to B &\implies \Cost{A}\geq \Cost{B}, \\
    A \to B &\implies \Yield{A}\geq \Yield{B}.
\end{align*}
We also get a sufficient condition based on the cost and yield of resources:
$$ \Yield{A} > \Cost{B} \implies A \to B.$$
These conditions correspond to the necessary and  sufficient conditions obtained in~\cite{Lieb2013} for thermodynamics.

\subsubsection{Tight currencies}

For any resource $A$ in the target, we can show  (Proposition~\ref{prop:cost_bigger_yield}) that
$$ \Cost{A}\geq\Yield{A} .$$
We can then look at the special case when the cost and the yield of a resource $A$ are equal, 
$\Cost{A}=\Yield{A}$. 
When they are furthermore achievable by a currency resource $C$, that is, when 
both $ C \to A $ and $ A \to C $ 
for some $C\in\mathcal C$, we call the currency \emph{tight} for the resource $A$. We will then also denote the set of such resources for which the currency is tight by $\mathcal S_=$.

We show in Theorem~\ref{thm:tightness_order} that for the set $\mathcal S_=$, the currency in fact gives a simple necessary and sufficient condition for resource transformations,
$$ A \to B \iff \Cost{A}\geq \Cost{B} $$
when $A,B\in\mathcal S_=$ (and of course similarly for the yield). This then also means that on the set $\mathcal S_=$, the resource theory provides a full order (up to equivalences). In particular, currency resources are also in $\mathcal S_=$ if they are in the target, and in turn, resources in $\mathcal S_=$ could always be added to the currency without changing the theory. 

Note that most quantum resource theories do not have a known tight currency for the whole target (that would have made things too easy). However, an example for tightness is given by the resource theory of noisy operations (or unital maps) when restricting to uniform states on some support (that is, states for which all non-zero eigenvalues are equal) \cite{Weilenmann2015}. 
Similarly, assuming that any state transformation in macroscopic thermodynamics can be implemented reversibly, work in equilibrium thermodynamics can be considered a tight currency \cite{Lieb1999,Lieb2003,Lieb2014}.

\subsection{Stage II: independent currency}

In addition to defining the cost and yield monotones on resources, it would be useful to use the currency to determine the \emph{balance} of general resource transformations: if resource $A$ is more valuable than  $B$, how much currency can be extracted in the process $A\to B$? Or how much currency needs to be supplied to transform $B$ into $A$?

\subsubsection{Independence between currency and target}

In order to answer these questions and formulate concepts such as ``adding a currency $C$ to a resource $A$'', we need to introduce a notion of \emph{composing resources}. For now we represent an arbitrary notion of composition simply as $(A,C)$; 
in Section~\ref{sec:technical}
we show that there is always a natural way to formalize it.
It is also essential to be able to address the currency and the target resources independently. For this it might help to think that we keep the currency in an independent \emph{wallet}, but we will see that this notion is more general in in Section~\ref{sec:technical}. 

In quantum theory, for example, this is guaranteed if the currency and the target resources live in different subsystems and are composed in tensor product, and the allowed transformations in the theory allow us to address their respective degrees of freedom individually. 
In real life, this independence is not a given, and is often an approximation: for example, in an optics experiment we might not be able to individually address an atom or transition without slightly disturbing the others.

We can formulate the independence condition as 
$$ C \to C' \implies (A,C) \to (A,C') $$
for all target resources $A$ and currencies $C$. 
This condition guarantees that a higher valued currency resource is really more powerful than a lower valued one in facilitating state transformations (Proposition~\ref{prop:value_operational_trafos}).
In Section~\ref{sec:technical} we show how to apply this conditions to realistic settings, where for example slight correlations between target and currency cannot be ruled out.

\subsubsection{Balance of resource conversions}

Composition of resources  allows us to define the \emph{balance} of a transformation from resource $A$ to resource $B$ in the target, conditioned on available currency $C$ as
\begin{align*}
    &\BalanceII{A}{B}{C} \\
    &= \sup_{C'\in\mathcal{C}} (\val(C')-\val(C): (A,C)\to (B,C')).
\end{align*}
In particular, if the balance is negative, it corresponds to minimal cost of transforming $A$ into $B$ given access to a currency  $C$. Of course, this quantity is only defined if the value of $C$ is actually large enough to afford the transformation. The independence condition ensures that the balance is meaningful, because it guarantees that the value of currencies is connected to how helpful  they are at facilitating transitions in the target (Proposition~\ref{prop:value_operational_trafos}).\footnote{The symmetric independence condition $A \to B \implies  (A,C) \to (B,C)$ would guarantee that $A\to B \implies \BalanceII{A}{B}{C}\geq 0 $, but does not seem to have a big impact otherwise.}

Note that in general this definition will depend on the available currency resource $C$ at the start of the transaction. This dependency could in principle be arbitrary: 
on the one hand, having access to additional currency in the wallet might facilitate the transformation and act partially like a catalyst, such that the balance of a transformation becomes larger the more currency is used to implement the transformation. In real life, for example, a client with more money may receive a special discount for a transaction (e.g.\ offered in the hope of acquiring a returning, rich loyal customer). 
On the other hand, however, having additional money might occasionally make the transaction more expensive: if instead of paying the exact amount of coins, one tries to pay with a larger note, one might end up paying more for an object or service (e.g.\ if the selling party cannot  give change).

As a last remark, note that we could have defined an analogous notion of balance,
\begin{align*}
    &\operatorname{Balance}^*(A\to B\vert C) \\
    &= \sup_{C'\in\mathcal{C}} (\val(C)-\val(C'): (A,C')\to (B,C)),
\end{align*}
where instead of conditioning on a starting resource $C$ in the currency, we demand that a final resource $C$ is reached at the end of the transformation. While we work with the first definition of balance in this paper, all the results could also be formulated with respect to this adapted notion.

\subsection{Stage III: fair currency}

In Stage II, the balance of a resource transformation can in principle depend on the available currency $C$ in the wallet. If that is not the case, we say that the currency is \emph{fair}, a property composed of two aspects.
On the one hand, it implies that there are no discounts for the poor, colloquially speaking: a transformation does not become cheaper to implement just because one has \emph{less} available currency. 
On the other hand, having access to \emph{more} currency does not make a transformation cheaper either | in other words, the currency does not act as a catalyst. We analyse the two conditions and their implications separately. 
In both cases the starting point is the same: 
suppose that we have resources $A$ and 
$B$ in the target,
and we can perform the transformation 
$$ (A,C_1) \to (B,C_2) $$
for some currency $C_1,C_2 \in\mathcal{C}$, with
$$ \Delta := \val(C_2)-\val(C_1)  .$$ 
We will now see what may happen if we start from a different currency $C_3$.

\subsubsection{Having less does not help}

Fairness in this direction means that if $C_3$ is more valuable than $C_1$, we can always find a final $C_4 \in \mathcal C$ with the same difference $\val(C_4)-\val(C_3)= \Delta$ that achieves the transformation 
$ (A,C_3) \to (B,C_4) $
(and vice-versa: we can  first fix $C_4$ with $\val{C_4}\geq \val{C_2}$ and look for an appropriate $C_3$).\footnote{This condition is required  to hold within suitable bounds, so that we do not hit the top boundary of the currency (if it exists), that is $\val(C_3)< \sup_{C\in\mathcal C} \val C -\Delta$ .} 

Naturally, this condition implies that the balance of a transformation cannot  decrease given access to more currency, that is (Proposition~\ref{prop:balance_increasing}),
\begin{align*}
    &\val C \geq \val C' \\ 
    &\implies \BalanceII{A}{B}{C}\geq \BalanceII{A}{B}{C'} ,
\end{align*}

In many familiar resource theories, currency resources can be composed to yield resources of higher value, such as an increasing number of Bell pairs in LOCC or a higher number of coins in real life. When these individual resources can be acted upon without disturbing the others, we can always ignore some of the currency and bring it back after the transaction. In these cases, more currency can never make a transaction more expensive; if anything, it can be cheaper when having access to more. 
Formally, what guarantees fairness in this direction in such cases is a condition of independence, like the one we introduced between currency and target, but extended to individual subsystems within the currency.  
For example, this is implied when currency resources are seen as a collection of objects in the paradigm of symmetric monoidal categories 
(one can understand copies of Bell pairs in LOCC as an example, see e.g.~\cite{Fritz2015}).~\footnote{Then, transformations are defined on individual objects and when two objects $A,B$ are composed, local transformations $f,g$ can also be combined and applied to the composed object such that 
$(f(A),g(B)) = (f,g)(A,B) $ 
(see e.g.~\cite{Coecke2014}). 
Choosing one of the functions as the identity, this ensures that we can put extra currency on the side for the purpose of a transformation and re-introduce it again later.}
In resource-theoretical approaches to macro and microscopic thermodynamics  we also find this condition: it is expressed   through the 
\emph{composition} axiom in Ref.~\cite{Lieb2013}, and is assumed in some models for quantum thermodynamics where work is stored in pure states on a number of individual qubits~\cite{Horodecki2011}. In fact, because this assumption is so common in resource theories, it is rarely questioned or exposed in this way. However, we would like to work with more general currencies than tensor product copies of individual states, and so we make this aspect of fairness explicit.

\subsubsection{Having more does not help}

Fairness in the other direction  means that  if the initial currency $C_3$ is less valuable than $C_1$, 
 we can again always find a $C_4 \in \mathcal C$ with the same difference $\val(C_4)-\val(C_3)= \val(C_2)-\val(C_1)$ that achieves the transformation 
$ (A,C_3) \to (B,C_4) $
(and also vice-versa: we can  first fix $C_4$ with $\val{C_4}\leq \val{C_2}$ and look for an appropriate $C_3$).\footnote{Here, we assume that $\val(C_3)\geq \val(C_1)-\val(C_2)$ so that we can actually afford the transformation with $C_3$.}

Naturally, this condition implies that the balance of a transformation cannot increase given access to more currency, that is (Proposition~\ref{prop:balance_decreasing}),
\begin{align*}
    &\val C \leq \val C' \\ 
    &\implies \BalanceII{A}{B}{C}\geq \BalanceII{A}{B}{C'}.
\end{align*}

Operationally, this direction of fairness makes sure that transformations are not easier to implement just because one has access to more currency | that is, the currency does not act like a catalyst. This would be guaranteed for example in a theory in which one is always able to `borrow' extra currency for free. 
Since this is in general not the case (for example due to finite-size effects), fairness in this direction is more common to fail than in the other.

\subsubsection{Both directions}
In case fairness holds in both directions, it follows that the balance of a transition does not depend on the initial currency $C$,
\begin{align*}
    \BalanceII{A}{B}{C} 
    &=  \BalanceII{A}{B}{C'}  \\
    &=: \Balance{A}{B} 
\end{align*}
for any $C,C'$ within suitable boundaries (Proposition~\ref{prop:balance_unique}).

For this new single notion of balance, we can also show that (Theorem~\ref{thm:balance_cost_yield})
$$\Balance{A}{B}\gtrsim \Yield{A}-\Cost{B}.$$ 
For resources $A,B$ for which  $\mathcal{C}$ is tight, furthermore
$$\Balance{A}{B}=\Yield{A}-\Cost{B}.$$

In many familiar resource theories in physics,  fairness of the currency is either taken for granted or imposed as a fundamental restriction on good currencies.
In thermodynamics, for example, it appears under the name of translational invariance: transformations are not allowed to depend on the initial state of the energy storage system
~\cite{Skrzypczyk2014,Aberg2013,Lieb2014,Masanes2014,Alhambra2016}. 
These storage systems are modeled as harmonic oscillators that mimic a classical weight; work, or the balance of a transformation, is  counted as the the difference between the energy of initial and final states. 
The underlying assumption required is that the weight system has many evenly spread energy levels and is far from the ground state.

\subsection{Pathological cases}

For Stage I currencies, we have seen that all resources $A$ in the target satisfy
$\Cost{A}\leq \Yield{A}$. We can read this as an impossibility statement that says that we cannot increase the amount of currency for free in the process of buying and re-selling a resource $A$ --- this would collapse the order in the currency and render the theory trivial. 
In a Stage II currency, there might however exist resources for which
$\BalanceII{A}{A}{C}>0$ for some $ C\in\mathcal C$,
that is, resources $A$ which allow us to generate a little bit of currency for free if we have access to $A$ and some particular value $C$ in the currency. Then, $A$ would essentially act like a catalyst in a process on the currency that is otherwise forbidden. 

The existence of such a resource $A$ does not make the theory trivial because the condition $\BalanceII{A}{A}{C}>0$ depends on the exact currency $C$ we start with. 
For example, it could require a valuable currency resource $C$, or could hold only for a very expensive resource $A$ that is hard to buy in the first place. In real life, for example, a house that is rented to a reliable tenant generates rent money every month, but requires a large investment to start with.
However, if the currency is fair (in the direction that more does not help), we can show that such a pathological resource  would again collapse the order in the currency. Pathological cases are discussed in detail in Appendix~\ref{appendix:pathological}.

\section{Application to realistic resources}
\label{sec:technical}
Now we show how to apply those ideas to  explicit descriptions of arbitrary physical  resources. The following approach allows us to model approximate transformations and generalize the idea of composition to a natural concept that applies for example when correlations between subsystems cannot be excluded, or when the subsystem structure is not clear. 
For details on this framework, we refer to Ref.~\cite{DelRio2015}; for the purposes of this work the following summary will suffice.

\subsection{Setup for resource theories}
\label{sec:preliminaries}

\subsubsection{Realistic descriptions of resources}
We may see the 
state space $\Omega$ of a theory as the \emph{language} that an agent uses to describe resources: for example, in quantum theory $\Omega$ could be the set of all density operators over a global Hilbert space; in traditional thermodynamics the set of all distinguishable macrostates. 
Realistically, agents may not know the  exact state of a system, and may instead describe resources through more coarse descriptions, like the $\eps$-neighbourhood of a state (obtained after tomography) or the specification of a few measurement outcomes. Such \emph{specifications} are simply subsets of the state space, like  $\ball^\eps(\rho)$, composed of all states  that are compatible with the agent's knowledge. 
Specifications have a natural partial order: if a description $V$ is more specific than another, $W$, we simply have $V\subseteq W$, for example $\mathcal B^\eps(\rho) \subseteq \mathcal B ^{\eps + \delta }(\rho) $. Together, the subsets of $\Omega$ form the
\emph{specification space} $S^\Omega$ --- the space of realistic resources.

\subsubsection{Transformations}

Physical actions implemented by an agent transform resources into resources ($f: S^\Omega \to S^\Omega$), in such a way that if an agent is unsure about the underlying state, this uncertainty carries through the transformation. 
For example, if the agent knows that the state of a system is either $\rho$ or $\sigma$, then after applying a transformation $f$, her knowledge is updated as  
$ f(\{\rho\}) \cup f(\{\sigma\}) $. 
In general, transformations $f\in\cT$ need to act element-wise, that is for any specification $V\subseteq \Omega$,
$$    f(V) = \bigcup_{\nu\in V} f(\{\nu\}). $$
This formalism also allows us to model cases where there is uncertainty about the exact effect of a transformation: for example, after performing process tomography of an experimental procedure $f$, we may learn only that $f$ is in a neighbourhood of a trace preserving completely positive (TPCP) map $g$. We could model this procedure  as $f: \{\rho\} \mapsto \ball^\eps(g(\rho))$.

\subsubsection{Resource theories}

As discussed in the introduction, a  resource theory is defined by a set $\cT$ of allowed transformations available to the agent, which act on the space of resources --- the specification space $S^\Omega$ that describes the resources from the point of view of an agent.
The set of allowed transformations induces a pre-order $\to$ of accessibility on $S^\Omega$, which encodes whether or not a resource $V\in S^\Omega$ can be transformed into another resource $W$. Formally,
$$ V \to W \iff \exists f\in \cT \ \text{s.t.} \ f(V) \subseteq W. $$ 
This pre-order combines the action of the transformation and the natural partial order on sets: operationally, this means that forgetting information (going to a less specific description) is always allowed in the resource theory. 
For example, in a resource theory of unital maps\footnote{Unital maps are all TPCP maps that preserve the identity, $\E(\id)=\id$. See Section~\ref{sec:example} for more details on the resource theory of unital maps and on how our results apply to this theory} on qubits we would find that  we can always reach the specification of an $\eps-$neighbourhood of the maximally mixed state: 
$\{\ket{0}\bra{0} \} \to \ball^\eps(\id_2/2), $
since there is a transformation $f\in\cT$ that achieves 
$ f (\{ \ket{0}\bra{0} \}) = \{ \id_2/2\} $
and
the maximally mixed state is a more specific description than an $\eps$-ball around it, 
$\{\id_2/2\} \subseteq \ball^\eps (\id_2/2). $

\subsection{Insights}

Let us now introduce the main insights that this approach brings into the subject of quantifying resources. 
In Section~\ref{sec:nontechnical} we saw that both currencies and the target are sets of resources. In our formalism, this means that they are sets of specifications, $\mathcal C, \mathcal S \subseteq S^\Omega$. In the following we explore some examples.
Formal definitions, results and proofs can be found in Appendix~\ref{appendix:value}.

\subsubsection{Rough currencies and single-shot statements}

Specifications are particularly useful 
when a  theory has a known currency $\mathcal C_{\text{ideal}}$ that cannot be implemented experimentally --- it may rely on idealizations like pure or perfectly uncorrelated states, for example. The experimenter may instead have access to coarser specifications of resources, for example $\eps$-neighbourhoods of the currency states $\mathcal C_{\text{real}}=\{ \ball^\eps (\rho)\}_{\rho \in \mathcal C_{\text{ideal}}}$. 
If  $\mathcal C_{\text{real}}$ is ordered (up to equivalences), 
it automatically forms a currency for some target space. Otherwise (for example, if two elements overlap too much), we may remove or replace some of the elements with other accessible resources until we obtain an ordered set. 
Naturally, we would not expect $\mathcal C_{\text{real}}$ to be as powerful a currency as $\mathcal C_{\text{ideal}}$: it might not reach the entire state space (e.g.\ because we do not have access to pure states), or it might offer a coarser quantification (because we removed some elements). Nevertheless, we  can always identify a maximal target set of specifications for which $\mathcal C_{\text{real}}$ is universal. 
In this case the target would include the set of $\eps$-neighbourhoods of all states reached by $\mathcal C_{\text{ideal}}$ (if the theory is stable, e.g.\ linear, under these neighbourhoods, as explained in \cite[Section V]{DelRio2015}). Once we find the appropriate target for $\mathcal C_{\text{real}}$, we may use all the tools of currencies, like cost, yield and checking for fairness, which apply to transformations between specifications, not only between states. 

We can also address questions about single-shot transformations, of the sort `what is the cost of reaching a final state, if we allow for a small error tolerance?' The agent encodes that error tolerance in an operational notion of closeness on the state space (like the trace distance), and uses it to build specifications of $\eps$-balls, $V^\eps \supseteq V$ \cite[Section V]{DelRio2015}. 
It follows that $\Cost{V} \geq \Cost{V^\eps}$ and  $\Yield{V} \geq \Yield{V^\eps}$. The exact result will depend on the theory \cite{Faist2012, Horodecki2009}; in the upcoming example of unital maps, the cost is characterized by smooth entropy measures. Similarly, with a Stage II currency we can talk about the balance of transitions between two resources $V^\eps$ and $W^\eps$.  The same reasoning applies to any currency  and target made of arbitrary specifications of resources more generally than just for well-behaved $\eps$-neighbourhoods. 

\subsubsection{Local resources and currencies}
In our approach we start from a global theory and model local resources as specifications in a global space. Doing this allows us to go beyond the tensor product to combine local resources.
For example, given a Hilbert space $\hilbert_A \otimes \hilbert_B$, a marginal state $\rho_A$ is a compact description of the set of all global states compatible with that marginal,
$$\local \rho_A := \{ \sigma_{AB}: \tr_B (\sigma_{AB}) = \rho_A\} .$$
A general way to compose local resources is to combine these specifications, for example
$$\local \rho_A  \cap \local \tau_B =
\{\sigma_{AB}:  \sigma_A = \rho_A \wedge \sigma_B = \tau_B \}.$$
This way, we can treat genuinely local knowledge without imposing additional non-local assumptions --- for example, we do not assume that the local states in the product state $\rho_A \otimes \tau_B \in \local \rho_A  \cap \local \tau_B$. If we do have some additional knowledge about the strength of correlations, we can include it in the specification. For example, the knowledge that the mutual information between the two subsystems is at most $\eps$ can be expressed as a specification 
$$ I(A:B)_{\leq \eps} :=  \{\sigma_{AB}:  I(A:B)_\sigma \leq \eps\},$$
 and our global specification then becomes
$$\local \rho_A  \cap \local \tau_B \cap I(A:B)_{\leq \eps}.$$

Often a currency $\mathcal C$ and a target $\mathcal S$  consist of local resources lying in different subsystems. Consider for example the case of LOCC, where a standard currency consists of different numbers of copies of Bell pairs. 
These can be stored in a wallet system shared by two agents Alice and Bob: we
may think of a decomposition of the global Hilbert space as 
$$\hilbert_{\text{global}} = \underbrace{\left( \bigotimes_{i=1}^N ( A_i  \otimes B_i) \right)}_{\text{wallet}} \otimes  \underbrace{A' \otimes \tilde B'}_{\text{target}},$$
where each $A_i$ and $B_i$ is a qubit. The currency is made of specifications of a certain number $n$ of copies of Bell pairs in the wallet,
$\mathcal C = \{ {\Psi^n} \}_{n=1}^N,$
with 
${\Psi^n} = \{ \sigma_{\text{global}}: \sigma_{A_1 B_1 \dots A_n B_n} = \pure \psi^{\otimes n} \} ,$
and $\ket \psi = (\ket{00} + \ket{11} )/ \sqrt2$. 
Note that  the currency is naturally ordered as ${\Psi^{n}} \to {\Psi^m} $ for $n>m$ (as ${\Psi^n} \subset {\Psi^m}$, that is we can always go from more to less Bell pairs by forgetting some --- in other words, by discarding or tracing them out).
Teleportation implies that this currency is universal for a target as large as the wallet, $\log (\max (|A'|, |B'|) )\leq N$. In our language it means that we can pick the target $\mathcal S$ to be any set of specifications that are local in $A'\otimes B'$, for example $\mathcal S = \{\local{\rho_{A'B'}} \}_{\rho}$, with
$\local{\rho_{A'B'}}= \{ \sigma_{\text{global}}: \sigma_{A'B'} = \rho_{A'B'} \} $ (Appendix~\ref{appendix:example}). The cost and yield of target resources are the usual entanglement of formation and distillation, respectively \cite{Horodecki2009}.

\subsubsection{Independence without composition}

We promised that specifications allow us to define independence between a currency and a target  without the need to talk of subsystems or any traditional notion of composition. 
At heart, independence means that we can change one quantity (the currency) without affecting the other (the target), as discussed in Section~\ref{sec:nontechnical}. While this trivially holds in the neat case of a tensor product structure between states and subsystems, the notion is more general.

The first requirement for independence was  that all states of the currency can coexist with all states of the target --- that is, if we combine the knowledge $C$ about the currency and the knowledge $V$ about the target, the resulting specification $C \cap V$ contains at least one global state compatible with this knowledge. If on the other hand $C \cap V = \emptyset$, this tells us that the two specifications contained contradictory knowledge --- which could happen for example if currency and target were defined on the same degree of freedom, so we could not have both at once.
Therefore, the first condition for Stage II currencies is \emph{compatibility}: for all currency and target resources  $C \in \mathcal C $ and $V \in \mathcal S$, we should have $C\cap V \neq \emptyset$. We do not impose any extra subsystem structure on $\mathcal C$ and $\mathcal S$, which are otherwise just sets of specifications, nor do we require a formal operation of composition of `local' resources. 

The second key idea  is that we can manipulate the currency without disturbing the target. In our language, this means that if we can transform $C\to C'$, then we can also act on the combined knowledge of currency and target as $C \cap V \to C' \cap V $, 
for all target resources $V\in\mathcal S$. 
Naturally, the Bell pairs from the previous example  form a  Stage II currency (Proposition~\ref{prop:Bell_currency}). 
Another example of a stage II currency is work in quantum thermodynamics (Proposition~\ref{prop:work_currency}).
More generally, these conditions can be satisfied even when correlations between currency and target cannot be ruled out --- in quantum theory, they do not need to be in a tensor product. Indeed, local knowledge such as $\local{\rho_A} \cap \local{\sigma_B}$ is sufficient to ensure independence, and therefore to define the balance of transformations \cite[Section IV]{DelRio2015}.

\subsection{Example: unital maps}
\label{sec:example}

Let us look at the particular example of a resource theory in which the allowed operations are given by unital, completely positive trace preserving maps on quantum states~\cite{Mendl2008}. On the level of state transformations and for classical systems, this resource theory also coincides with the resource theory of noisy operations\footnote{The actual set of unital maps is a superset of the maps achievable by noisy operations.}~\cite{Birkhoff1946,Horodecki2003,Mendl2008,Faist2012,Gour2013}. As such, it characterizes a range of physical situations from an agent processing information in a noisy environment to thermodynamics on degenerate energy eigenstates.
On the level of quantum state transformations, the resource theory of unital maps is well understood: the pre-order on quantum states is given by majorization\footnote{For simplicity, we restrict our analysis to the case of same input and output dimensions.}~\cite{Hardy1952,Uhlmann1970,Ruch1970,Ruch1975,Ruch1976,Ruch1978JCP,Ruch1980JMAA,Horn1985,Joe1990,Bhatia1997,Nielsen2001,Marshall2011,Horodecki2003,Faist2012,Gour2013}, and we can quantify resources through smooth entropies~\cite{Weilenmann2015,Faist2012,Brandao2013b,Renner2004,Renner2005,Datta2009IEEE,Tomamichel2012}.

However, for specifications the situation is not so clear: how can we assign value to general resources $V \in S^\Omega$? Are there simple necessary and sufficient conditions for resource transformations  $V\to W$? 
Solving these questions would allow us to characterize the pre-order for realistic descriptions of resources, including approximations around quantum states and specifications of a few selected properties of the system. Furthermore, it would give us insight into how to define entropy of specifications, study how much work is needed to erase the information encoded in a specification (analogous to Landauer's principle~\cite{Landauer1961,Bennett1982,Faist2012}), and how to draw a link between microscopic and macroscopic thermodynamics, as outlined in Ref.~\cite{DelRio2015}.
Finally, the question of when a specification can be transformed into another bears resemblance to the question of the work cost of a general process~\cite{Faist2012} as well as to the question of when a set of states can be transformed into another set~\cite{Uhlmann1980,Chefles2004,Heinosaari2012,Huang2012}.

Here we introduce three alternative  currencies for the resource theory of unital maps on specifications.
First we define a Stage I currency  and compute the cost and yield of specifications, which already give us necessarily and sufficient conditions for resource transformations.
Then we give a similar Stage II currency, and discuss examples of alternative currencies that make use of specifications more explicitly. Finally, we study smooth transformations between $\eps$-neighbourhoods of states. Our choice of currency is inspired by Ref.~\cite{Weilenmann2015}, where entropies for quantum states were derived from similar considerations | our results in this section could hence be interpreted as a step towards understanding entropies for specifications.

\subsubsection{Stage I currency}

Let the global state space $\Omega$ of the theory be all quantum states on a $d$-dimensional Hilbert space. Specifications $V\in S^\Omega$ are then sets of such states.
We can define a universal currency $\mathcal C = \{C^k \}_{k=1}^ d \cup \{\Omega\}$  as the set of 
uniform states of different ranks, 
$$ C^k = \left\{\frac{1}{k} \sum_{i=1}^k \pure{i}\right\} = \left\{\frac{\Pi^{k}}{k}\right\},$$
where $\Pi^{k}$ denotes the projector onto the first $k$ eigenstates of a fixed basis. 
We take the value function
$$ \val\left(C^k\right) = \log (d) - \log k,$$
such that states with lower rank are more valuable.\footnote{Note that $\val(\Omega)= \val(C^ d) =0$, since the two are interconvertible under unital operations.}
We prove that  $\mathcal C$ indeed forms a currency for the global target $\mathcal S = S^\Omega$ in Proposition~\ref{prop:currency_unital}. The cost of a specification $V$ in terms of $\mathcal{C}$ is then given by
\begin{align*} 
\Cost{V} 
&= \log d - \sup_{\rho\in V} \log \lfloor 2^{H_\text{min}(\rho)}\rfloor
\end{align*}
where $\hmin$ is the min-entropy~\cite{Renner2005,Tomamichel2012,Tomamichel2016} defined as $\hmin(\rho) = - \log \lambda_\text{max}(\rho) $, 
with  $\lambda_\text{max}(\rho) $ being the largest eigenvalue of $\rho$,
and $\lfloor \cdot \rfloor$ the nearest integer (from below) to the enclosed expression (Proposition~\ref{prop:cost_majorization}).
This result can be seen as an analogue to the lower bound for entropy of quantum states $\rho$ found in Ref.~\cite{Weilenmann2015}, which obtained exactly $H_\text{min}(\rho)$ as a monotone. Here, we see that more generally there is an optimization over the states in the specification.

Note that as we evaluate the cost, we  need to round down 
$ \lambda_\text{max}^{-1}(\rho) $, because 
there is a limit to how well we can approximate  $\lambda_\text{max}$ by means of a rational $\frac{1}{n}$. For example, in a two-level system, our currency would only contain elements of two different values: 1 (pure state), and 0  (fully mixed state). Then, the cost of a single-qubit specification is either 0 (for all specifications containing the fully mixed state) or 1 (for all others), but nothing in between.
This is a finite-size effect that is not critical for large systems; it nevertheless tells us that our choice of currency, while universal, is not very fine-grained. An alternative currency for single qubits could be for example  $\mathcal C' =\{C_p\}_{p=0}^{1/2} \cup \{\Omega\}$, with $C_p =\{ p \, \pure0 +(1-p)\, \pure1\}$, which would recover $\operatorname{Cost}(V)_{\mathcal C'}=\log d - \sup_{\rho\in V} \hmin (\rho)$.\footnote{In Ref.~\cite{Weilenmann2015}, this issue was circumvented by extending the state space from density matrices to continuous step functions. On the level of state transformations, the theory of unital maps (or noisy operations) boils down to a simple majorization condition, which can easily be extended to such step functions. However, here we would like to be a little more careful, especially since we deal with specifications $V$ in general, for which the pre-order structure is more complicated and not fully characterized by majorization.}

As for the yield of specifications in terms of $\mathcal C$, let us first look at the single-qubit example.
Suppose that we start from the specification 
$ V = \{ \pure{0}, \pure{1} \} $. 
While for any individual pure state the yield is $1$ since we can always generate the pure currency state from it, the same is not true for the specification $V$. The reason is that there is no \emph{single} protocol that achieves the pure currency state for \emph{both} states  $\pure{0}$ and $\pure{1}$, so that it could be applied without knowing which one is actually the case. Instead, in the case of $V$ one can only generate the fully mixed currency state, which can be done with a single process that works for both cases (the discard-and-prepare map $\E(\rho) =\id/2$, which is unital). The yield of $V$ is hence $0$ | the same as the yield for a general convex mixture of the two pure states in $V$.

In general, we find that the yield of a specification in terms of $\mathcal C$  is given by
\begin{align*} 
\Yield{V} 
&= \log d - \max_{\rho \in V^\P}  H_0 (\rho),
\end{align*}
where $V^\P$ denotes the convex hull of $V$, which consists of all the convex (probabilistic) mixtures of the states in $V$~\cite[Section VI]{DelRio2015}) and $H_0(\rho)$ is the order zero quantum R\'enyi entropy defined as $H_0(\rho) = \log \operatorname{rank} (\rho)$~\cite{Renner2005,Tomamichel2012,Tomamichel2016}. 
This is proven in  Proposition~\ref{prop:yield_majorization}.

\subsubsection{Stage II currency}

We can now analyse an independent currency that lives in a different subsystem to the target specifications.
Consider a global state space of density matrices on two systems $W\otimes S$, 
where $W$ will be our wallet and system $S$ the target, with dimensions $d_W \geq d_S$.
More precisely, we can show that 
the set
$ \mathcal{C} 
= \{ C^k \}_{k=1}^{d_W} \cup \{ \Omega \} $,
with
$$ C^k = \left\{ \sigma_{WS}: \tr_S(\sigma_{WS}) = \frac{\Pi^{k}}{k} \right \} $$
forms a currency for the target $\mathcal S$ defined through
$$ V \in \mathcal S \iff V = \{ \sigma_{WS}: \tr_C(\sigma_{WS}) \in V_S \} $$
for some set
$ V_S $ of density matrices on system $S$ (Proposition~\ref{prop:currency_unitalII}). This currency furthermore satisfies \emph{independence} according to Definition~\ref{def:independence}, even if  most elements of the currency are not in a tensor product with target resources. 

We can easily see that the currency is not fair, because of the finite-size effects discussed previously, but in the limit of large wallet dimension $d_W$, we can show that
$$ \Cost{V} \approx \log d_S - \sup_{\rho_S\in V_S} \log \hmin (\rho_S) ,$$
to an arbitrarily good approximation (Theorem~\ref{thm:cost_yield_unitalII}).\footnote{This issue has also been discussed in~\cite{Weilenmann2015}, where it was noted that access to large ancilla systems can help approximate the step functions arbitrarily well.}

\subsubsection{Alternative currency}
As we saw, there can be many options for a currency within a resource theory, and the ultimate choice is up to the user (for example, whether they will be treating small or large systems, or whether they can easily distinguish states that are close). We would like to finalize this example with a currency that makes more explicit use of specifications  to formalize lack of knowledge. The idea is to model an agent who knows only  that the system could be in any mixture of $k$ pure states ---  weaker than knowing that the system is in a uniform mixture of those $k$ states.
Consider $ \mathcal C_2 = \{C_2^k\}_k \cup \{\Omega\} $,
with 
$$ C_2^k = \left(\bigcup_{1\leq i \leq k} \{ \ket{i}\bra{i} \}\right)^\P ,$$ 
where instead of taking maximally mixed states of different ranks like $\mathcal C$ before, we take the convex hull of differently sized sets of orthogonal pure states on the currency system.

In comparison to our original currency $\mathcal C$, note that for each $k$, $C^k \subseteq \mathcal C^k_2$.
The set $\mathcal C_2$ forms again a currency for the target $S^\Omega$, and is in fact equivalent to the currency $\mathcal C$ before since, under unital operations,
$C^k \to C^k_2 $  and $ C^k_2 \to C^k $
for all $k$ . In particular, we  can take the usual  value function
$\val (C_2^k) = \log d - \log k $ and obtain the same expressions for cost and yield.
The convex hull in the expression for $C^k_2$ is crucial | without it, we would still get a currency but with different results for the cost and yield of resources (Proposition~\ref{prop:equivalent_currency}).

\subsubsection{Smooth transformations}

As a special case of our framework, we can model scenarios where we allow for some error probability in the output of a transformation, or where we need to make sure that our protocol is robust against errors in the initial resources.
To illustrate this, we can look at  resources that correspond to approximate quantum states, $\mathcal B^\eps(\rho)$, according to some metric like the trace distance or the purified distance based on the fidelity (more general approximations are discussed in Ref.~\cite{DelRio2015}). 
In this case, we recover the smooth-min entropy~\cite{Renner2005,Tomamichel2012,Tomamichel2016} as a monotone, as 
$$      \Cost{B^\eps(\rho)} 
= \log d - \log \lfloor 2^{\hmin^\eps(\rho)}\rfloor.$$
On the other hand, smoothing the input does not always give us very meaningful results for this choice of currency, as for example $\Yield{B^\eps(\rho)} =0$ for all $\eps >0$. 
The usual bound for the balance of resource transformations in terms of cost and yield  (which works by selling the initial resource and buying the final one)  is not helpful when the yield is zero: for example,
the most efficient way to go from $B^\eps(\rho)$ to $B^{\eps'}(\rho)$ will never go through this currency.
This effect happens partly because the currency is relatively course-grained (it contains only flat states), and also because it consists of exact states rather than smooth specifications (unlike the previously discussed $\mathcal C_{\text{real}}$). By smoothing the currency or adopting a version of yield that smooths over the output, we could arrive at more useful statements, involving for example the smooth max-entropy. First steps on smooth transformations can be found in Appendix~\ref{appendix:example}.

\section*{Conclusions}
\label{sec:conclusions}

While each resource theory has its own understanding of how to identify precious resources, there are common tools to address this question. 
In this work, we abstract from particular examples and identify the general concept of \emph{currencies} that can be used to quantify the cost of resource transformations as well as the value of particular resources. 
Our approach allows us to go beyond commonly employed assumptions of perfect independence between the currency and the system of interest, exactly known quantum states, perfect and uncorrelated copies of states on individual subsystems, or any special form or properties of the currency states such as the scalability of equilibrium systems employed in Refs.~\cite{Lieb1999,Lieb2003,Lieb2013,Lieb2014}. As such, our results seem particularly relevant for thermodynamics of macroscopic systems or systems of limited control, where slight correlations cannot be ruled out or only a few properties of the systems of interest are known. Furthermore, our approach could prove powerful in the context of adversarial settings in cryptography, in which worst-case scenarios need to be assumed~\cite{Landauer1993, Morgan2015}.

Using the formalism of specifications developed in~\cite{DelRio2015}, we can quantify the value of arbitrary descriptions of resources, paving the way to finding entropic  quantities that characterize non-probabilistic states of knowledge.

As we saw in the example of the resource theory of unital maps, there are many possible choices of currency within a theory. Traditional choices such as copies of pure states or Bell pairs are useful to find monotones from asymptotic conversion rates 
\cite{Fritz2015, Brandao2011a}. On the other hand, in realistic settings with limited resources, these strict currencies may result in a big divergence between cost and yield of target resources, as happened for resources that were $\eps$-neighbourhoods of states. As a consequence of this gap, the resulting bounds on the balance of resource conversions become very loose, and do not give a very useful characterization of the pre-order of the resource statement. 
A commonly used way out is to always work in the regime of large systems, where these problems become less significant; however, we believe it is more natural to stay in the experimentally realistic setting and adapt our theoretical  tools to meet the constraints in the lab. 
In order to find tighter and meaningful bounds,  we might for example use finer currencies, or smooth currencies  made of specifications (e.g.\ currencies that are themselves  $\eps$-neighbourhoods of states). Alternatively we might keep the currency but  relax our definitions of cost, yield and balance to allow for smoothing on the input and output, corresponding to a small error probability.

\subsection*{Related work}

In abstracting from particular resource theories and identifying common tools to quantifying resources, our work is similar in spirit to Refs.~\cite{Coecke2014,Fritz2015,Brandao2015}. However, while these works also explore the mathematical structure of resource theories, they study the case where resources have a clear local structure, equipped  with a composition operation that ensures independence between currency and target (this is the case for symmetric monoidal categories \cite{Coecke2014,Fritz2015} and in particular the tensor product of quantum states \cite{Brandao2015}). In other words, those works study Stage II currencies, with one direction of fairness (having less does not help) guaranteed. 
We see our approach as complementary to Fritz's \cite{Fritz2015}, in the sense that he explores in depth how much we can do once we have a definite notion of composition, while we ask the broader question of what independence and composition mean, and what  we can still do when  we cannot guarantee that local structure in the space of resources.
Another difference is  our interest in  single-shot settings, compared to a stronger focus on conversion rates and asymptotic scenarios in Refs.~\cite{Coecke2014,Fritz2015,Brandao2015}.

We drew inspiration from many particular examples of resource theories, and especially from the abstract approach to thermodynamics in Lieb and Ynganvon's work~\cite{Lieb1999,Lieb2003,Lieb2013,Lieb2014}. 
We generalize their ideas to arbitrary resource theories beyond thermodynamics. In addition, we weaken and clarify the  assumptions needed to recover the monotones; in particular, we can avoid assuming any additional structure on the currency resources, such as the scalability requirement present in those works. 
Not surprisingly, our results resemble the quantities and considerations of Refs.~\cite{Lieb1999,Lieb2003,Lieb2013,Lieb2014, Weilenmann2015}, as well as those of approaches to quantum thermodynamics~\cite{Horodecki2011,Brandao2013b,Gour2013}. Our results, however, have a wider range of application.

Our work also bears resemblance to the article by Gallego \emph{et al.}\ on defining work in quantum thermodynamics from operational principles~\cite{Gallego2015}. Their paper  studies how a set of axioms motivated by the second law of thermodynamics influences possible definitions of cost or balance of resource transformations. 
The setting of two players Arthur and Merlin in that paper relate very closely to our splitting of resources into currency and target. Furthermore,  the definitions and  properties of the resulting work function are similar to our statements about balance. For example, the total order on the currency resources that we demand seems to follow from the axioms in Ref.~\cite{Gallego2015} | their approach can hence be seen as an interesting complement to ours. However, we go beyond the scope of their results not only since we consider general resource theories beyond quantum thermodynamics, but also because we do not limit ourselves to exact quantum states and tensor product composition. Furthermore, we discuss a range of properties of currencies and their implications, beyond seeking a unique notion of balance.

\subsection*{Directions}

It would be interesting to analyse more thoroughly the problem of general resource conversion in terms of monotones or specific transformation criteria. In particular, it would be key to generalize familiar monotones such as the relative entropy to the fixed point or free resources of the resource theory \cite{Brandao2015} to the case of specifications. Similarly, one could try to generalize known results from state to state transformations for particular resource theories, such as the conditions of majorization or thermo-majorization in noisy and thermal operations.  

Another direction of further research would be to look more closely at resource theories in which we can only find currencies for a limited target.
A good example of this is the resource theory of thermal operations~\cite{Janzing2000,Brandao2011a}, in which work can be considered a currency only for states that are block-diagonal in the energy eigenbasis. Since the theory is time-translation symmetric, we cannot introduce coherence from purely incoherent resources~\cite{Lostaglio2015, Winter2015, Baumgratz2013}. 
This problem could perhaps be solved by introducing two or more different currencies that function together and may or not be traded against each other (similar to recent ideas on non-commuting conserved charges in thermodynamics~\cite{Halpern2015, Guryanova2015, Lostaglio2015b}) | for example, by introducing both a work and a coherence storage system. It would be interesting to extend our results to such cases and 
formulate the features of such a theory in our general language.  Finally, it would be interesting to characterize the conditions under which currencies become fair, like the regime of large systems, and the connection to conversion rates.


\subsection*{Acknowledgements}
We thank Tobias Fritz for pointing out the complementary definition to balance and bringing up issues with pathological cases, Philipp Kammerlander for translating Ref.~\cite{Gallego2015} and for comments on this manuscript, 
Sandu Popescu for inciting us to scrap off most of the formalism,
Renato Renner for encouraging us to focus on the essential properties of currencies, and Cyril Stark for canny convexity consultancy.

LK acknowledges support from  the European Research Council via grant No. 258932, the Swiss
National Science Foundation through the National Centre of
Competence in Research \emph{Quantum Science and Technology}
(QSIT),  and   the  European  Commission  via  the  project \emph{RAQUEL}. 
LdR acknowledges support from ERC AdG NLST and EPSRC grant \emph{DIQIP}, and from the
Perimeter Institute for Theoretical Physics. Research at Perimeter Institute is supported by the government of Canada through Industry Canada and by the Province of Ontario through the Ministry of Economic Development \& Innovation. This work was partially supported by the COST Action MP1209. 

\onecolumngrid
\appendix

\addcontentsline{toc}{section}{\sc{Appendix}}

\section{Formalizing currencies}
\label{appendix:value}

Here we define
currencies in general resource theories using the 
formalism of~\cite{DelRio2015}
for simplicity and generality,
and prove the general results about currencies laid out in the main text. 
These notions could also be written in a more traditional
language of resource theories, as in Section \ref{sec:nontechnical}.

\subsection{Stage I: currencies as a universal `standard'}

Given a global state space $\Omega$, following~\cite{DelRio2015} and Section~\ref{sec:preliminaries}, resources are specifications, that is sets $V\subseteq \Omega$. We call the set of all non-empty subsets of $\Omega$ the specification space $S^\Omega$. 
A currency is then a subset $\mathcal{C}\subseteq S^\Omega$ of such resources that is ordered (up to equivalences) and universal for a target set $\mathcal{S}\subseteq S^\Omega$ of resources, as explained in Section~\ref{sec:nontechnical}.  
For simplicity, we furthermore demand that the trivial (free) resource $\Omega$ is both in the currency and the target, $\Omega\in\mathcal{C}$ and $\Omega\in\mathcal{S}$.

\begin{definition}[Currency]
\label{def:currency}
Let $(S^\Omega,\cT)$ be a resource theory.
A subset of specifications $\mathcal C \subseteq S^\Omega$ forms a \emph{currency} for a subset $\mathcal{S}\subseteq S^\Omega$, called the \emph{target}, if it satisfies both
\begin{enumerate}

\item {\bf Order.}  $\mathcal C$ is ordered (up to equivalences), that is for  all $C,C'\in\mathcal{C}$, either $C\to C'$ or $C'\to C$, and $\Omega \in \mathcal C$.

\item {\bf Universality.} for all target specifications $V \in \mathcal{S}$, there exists an element of the currency $C \in \mathcal C$ such that 
$ C \to V $,  and a currency element $C'\in\mathcal C$ such that
$ V \to C'$, and $\Omega \in\mathcal S$.

\end{enumerate}
\end{definition}

Since the set of currency resources is ordered, we can define the \emph{value} of the currency via a function to the set of positive real numbers.\footnote{Provided the order on $\mathcal{C}$ fits into this set.}

\begin{definition}[Value of currency]
\label{def:value}
We define the \emph{value} of the currency via a monotonic  function
$\val: \mathcal{C}\to \mathbb{R}_+ $,
that is,
$$\val (C')\geq\val(C)\iff C'\to C$$
for any $C,C'\in\mathcal{C}$.
As convention, we define the following threshold values:

\begin{enumerate}
\item a minimum value $c_{\min}$ which we set to zero, $\val (\Omega):=c_{\min}=0$,
\item a \emph{saturating} value relative to a target $\mathcal S$,
\begin{align*}
c_\text{sat}=\sup_{V\in \mathcal{S}} \ \Cost{V},
\end{align*}
\item the \emph{supremum} value of a currency (which may be larger than $c_\text{sat}$, and in particular may be infinite),
\begin{align*}
c_\text{sup}=\sup_\mathcal{C} \ \val(C)
\end{align*}
\end{enumerate}

\end{definition}

Naturally, it follows that $\val (C')\geq\val(C)$ implies that $C'$ is at least as powerful as $C$, in the following sense.

\begin{restatable}[Value is operational]{proposition}{propValueOperational}
\label{prop:value_operational}
For any $C,C'\in\mathcal{C}$ such that $\val(C)\geq\val(C')$ and any $A\in\mathcal{S}$,
it holds that
\begin{align*}
    C' \to A &\implies C \to A \\
    A \to C &\implies A \to C'.
\end{align*}
\end{restatable}
\begin{proof}
This result follows straightforwardly from $\val(C)\geq \val (C')$, which implies that $C \to C'$, together with the transitivity of $\to$, since we can concatenate individual processes to get
\begin{align*}
    C \to C', C' \to A &\implies C \to C' \to A \\
    A \to C, C \to C' &\implies A \to C \to C'.
\end{align*}
\end{proof}

We can now define the \emph{cost} and \emph{yield} of resources as follows.

\begin{definition}[Cost and yield of resources]
\label{def:cost_yield}
Let $(S^\Omega,\cT)$ be a resource theory with the currency $\mathcal{C}$ for the target $\mathcal S$. 
The \emph{cost} and \emph{yield} of a target resource $V\in\mathcal{S}$ in terms of the currency  $\mathcal{C}$ are defined as
\begin{align*}
\Cost{V} &:= \inf_{C\in\mathcal C} (\val C: C \to V),\\
\Yield{V} &:= \sup_{C\in\mathcal C} (\val C: V\to C).
\end{align*}
\end{definition}

The following theorem shows why currencies are so useful in resource theories. Namely, they allow us to derive necessary and sufficient conditions for resource conversion. 
With this theorem, we also recover the results of~\cite{Lieb2003, Lieb2014} more generally for all resource theories equipped with a currency.

\begin{restatable}[Conditions for resource conversion in terms of a currency]{theorem}{thmCurrencyConditions}
\label{thm:currency_conditions}
Let $(S^\Omega,\cT)$ be a resource theory with a currency $\mathcal{C}$ for a target $\mathcal S$.  Let  $V,W\in \mathcal S$. 
Then
\begin{enumerate}
\item if  $\Yield{V} > \Cost{W}$ then  $V \to W$, and

\item Cost and Yield are monotones, that is, if $V \to W$, then 
\begin{align*} 
\Cost{V} &\geq \Cost{W} \\
\Yield{V} &\geq \Yield{W}.
\end{align*}
\end{enumerate}
\end{restatable}
\begin{proof}
For the first statement, note that 
$\Yield{V} > \Cost{W}$ implies that 
there are some $C_y,C_c \in \mathcal C$ with 
$$ \Yield{V}\geq \val(C_y) > \val(C_c) \geq \Cost{V} $$
such that
$ V \to C_y $ and 
$ C_c \to V $ 
(they are achievable selling and buying currency values). 
Then we can compose the processes
\begin{align*}
V \to C_y \to C_c\to W.
\end{align*}
To show monotonicity of cost and yield, note that we can again employ transitivity to get that, for any $C,C'\in \mathcal C$, if $V\to W$ then
\begin{align*}
    C \to V &\implies C \to W \text{ and}\\
    W \to C' &\implies V \to C',
\end{align*}
and so the statement follows directly from the definitions of Cost and Yield as
\begin{align*}
\Cost{V} &:= \inf_{C\in\mathcal C} (\val C: C \to V),\\
\Yield{V} &:= \sup_{C\in\mathcal C} (\val C: V\to C).
\end{align*}
in Definition~\ref{def:cost_yield}.
\end{proof}

We can then also show the following.
\begin{restatable}[Cost is bigger than yield]{proposition}{propCostBiggerYield}
\label{prop:cost_bigger_yield}
Let $(S^\Omega,\cT)$ be a resource theory with a currency $\mathcal{C}$ for a target $\mathcal S$. Then, for any resource $V\in\mathcal S$,
$$ \Cost{V}\geq\Yield{V} .$$
\end{restatable}
\begin{proof}
Suppose that $\Yield{V}>\Cost{V}$. Then we can find currency resources $C_y,C_c\in\mathcal C$ with
$$ \Yield{V}\geq \val (C_y) > \val (C_c) \geq \Cost{V} $$
such that
$ V \to C_y $
and 
$ C_c \to V .$
But then we can compose the processes to get 
$ C_c \to V \to C_y $
and so due to monotonicity of value we would need
$ \val (C_c )\geq \val (C_y ),$
contradicting our initial assumption that $\Yield{V}> \Cost{V}$.
\end{proof}

We now give a formal definition of \emph{tightness} and show that when considering resources for which the currency is tight, Theorem~\ref{thm:currency_conditions} becomes much stronger: the monotonicity of  cost and  yield become necessary and sufficient conditions for resource transformations.

\begin{definition}[Tightness]
A currency $C$ with target $\mathcal S$ is called \emph{tight} for a 
resource $V\in\mathcal S$ if 
$\Cost{V} = \Yield{V}$.
We then denote the set of resources for which the currency is tight by $\mathcal S_= \subseteq \mathcal S $.
\end{definition}

\begin{theorem}[Tightness yields necessary and sufficient condition]
\label{thm:tightness_order}
For two resources $V,W\in\mathcal S_=$ for which a currency $\mathcal C$ is tight, 
$$ V \to W \iff \Cost{V} \geq \Cost{W} . $$
The pre-order on the set $\mathcal S_=$ then becomes a full order (up to equivalences), that is, 
$$ \forall\, V,W \in \mathcal S_=,  \text{ either } V \to W \text{ or } W \to V .$$
\end{theorem}

\begin{proof}
This statement follows directly from $\Cost{V} = \Yield{V}$ and $\Cost{W} = \Yield{W}$ together with the proof of Theorem \ref{thm:currency_conditions}, where we can see that in the case when the cost and yield are achievable, $\Yield{V}\geq \Cost{W}$ suffices as a sufficient condition for $V\to W$.

The full order (up to equivalences) on $\mathcal S_=$ then follows directly from the full order on $\mathbb R_+ $ to which the cost maps.
\end{proof}

\subsection{Stage II: independent currency}

In Stage II, we now look at currencies that satisfy also the \emph{independence} property, namely that the currency resources are compatible with any state on the target, and that the currency can be transformed individually without disturbing the target.

\begin{definition}[Currency independent of target]
\label{def:independence}
Let $(S^\Omega,\cT)$ be a resource theory equipped with a currency $\mathcal C \subset S^\Omega$ for a target $\mathcal S$. Then the currency $\mathcal{C}$ is called \emph{independent} of the target if 
\begin{enumerate}

\item all specifications $C \in\mathcal{C}$ are \emph{compatible} with all specifications $V\in\mathcal S$, i.e.\ $ C \cap V \neq \emptyset $, and

\item if we can transform $C\to C'$, then we can do it without disturbing the target, that is, 
$$ C \to C' \implies C \cap V \to C' \cap V $$
for all $V\in\mathcal S$. 

\end{enumerate}
\end{definition}

The following proposition shows how a currency resource with higher value can implement all the resource transformations that a lower currency value can achieve. It extends Proposition~\ref{prop:value_operational} and implies its result by setting $V,C''=\Omega$ or $ W,C''=\Omega$ respectively.

\begin{restatable}[Value is operational (independent currencies)]{proposition}{propCurrencyValue}
\label{prop:value_operational_trafos}
Let $(S^\Omega,\cT)$ be a resource theory with a currency $\mathcal{C}$ independent of target $\mathcal S$. Then, for $C,C',C''\in \mathcal{C}$ and $V,W\in \mathcal S$, if $\val(C') \geq \val  (C)$, then
\begin{align*}
    V \cap C\to W \cap C'' &\implies V\cap C' \to W \cap C'' \\
    V \cap C''\to W \cap C' &\implies V\cap C'' \to W \cap C.
\end{align*}
\end{restatable}
\begin{proof}
To show this statement, we employ transitivity together with independence. Independence means that
$ \val(C') \geq \val(C) \implies V \cap C' \to V \cap C $ for all $V\in\mathcal S$, and so we find that
\begin{align*}
    V \cap C \to W \cap C'' &\implies V \cap C' \to V \cap C \to W\cap C'', \\
    V \cap C'' \to W \cap C' &\implies V \cap C'' \to W \cap C' \to W\cap C,
\end{align*}
which guarantees the result due to transitivity.
\end{proof}

We now define the \emph{balance} of resource transformations conditioned on the available currency $C$ in the wallet as follows. Naturally, this is  only defined when the transformation is actually possible with $C$.

\begin{definition}[Balance]
\label{def:balance}
Let $(S^\Omega,\cT)$ be a resource theory with the currency $\mathcal{C}$ for a target $\mathcal S$ that satisfies the \emph{independence} condition (Definition~\ref{def:independence}). 
Then we define the \emph{balance} of transforming a resource $V$ into another resource $W$ on the target when having access to the currency resource $C\in\mathcal C$ as
\begin{align*}
\BalanceII{V}{W}{C}:= \sup_{C'\in\mathcal{C}}\{\val(C')-\val(C): \
V \cap C\to W\cap C'\}.
\end{align*}
\end{definition}

As we have discussed in Section~\ref{sec:nontechnical}, the balance will in general depend on the available currency $C$ in the wallet. If the currency is fair, there will not be such a dependence (see Proposition~\ref{prop:balance_unique}).

\subsection{Stage III: fair currency}

\begin{definition}[Fair currency]
\label{def:fairness}
A currency $\mathcal C$ with target $\mathcal S$ and value function $\val: \mathcal{C}\to\mathbb{R}_+$ is called \emph{fair} if for any $C_1, C_2 \in \mathcal C$ with
$ \Delta := \val (C_2) -  \val (C_1)  $
and $V,W\in \mathcal S$ such that 
$$V \cap C_1 \to W \cap  C_2, $$
the following holds:
\begin{enumerate}
    \item for any $C'_1\in \mathcal C$ with
    $ - \Delta \leq \val (C'_1)  < c_\text{sup} - \Delta $, 
    there exists a $C'_2\in\mathcal{C}$ such that
    $$ V \cap C'_1 \to W \cap  C'_2$$
    and
    $\val(C'_2)-\val(C'_1)=\Delta$.
    \item for any $C'_2\in\mathcal C$ with
    $ \Delta \leq \val(C'_2)$, 
    there exists a $C'_1\in\mathcal{C}$ such that
    $$ V \cap C'_1 \to W \cap  C'_2$$
    and
    $\val(C'_2)-\val(C'_1)=\Delta$.
\end{enumerate}
We say that a currency is \emph{good for the rich} (having less does not help) if the above two conditions hold for 
$\val (C'_1)\geq \val (C_1)$ and $\val (C'_2)\geq \val (C_2)$ 
respectively. 

Similarly, we call a currency \emph{good for the poor} (having more does not help) if they hold for $\val (C'_1)\leq \val (C_1)$ and $\val (C'_2)\leq \val (C_2)$.
\end{definition}

As explained in Section~\ref{sec:nontechnical}, fairness guarantees that, within suitable bounds, whether or not a transformation can be facilitated with currency in the wallet only depends on the difference in value between initial and final currency resource, but not on the absolute values of currency in the wallet.

Looking at the boundaries from fairness more precisely, we see that the left-hand ones,
$-\Delta\leq \val(C_1')$ and $\Delta\leq \val(C_2')$, simply makes sure that the value of $C_1'$ is not too low for the balance to be defined. 
The right-hand conditions concern large  currency values and transformations that allow us to extract extra currency. 
For example, if we start from a currency of high value $\val (C_1')\geq c_\text{sup} - \Delta$, the balance of the transformation must be less than $\Delta$, because we cannot exceed the boundary  $c_\text{sup}$ by definition.

\subsubsection{Having less does not help}

We will now discuss the implications of a currency that is good for the rich. For such a currency, having less does not help | there are no discounts for the poor.

\begin{proposition}[Balance non-decreasing  with available currency]
\label{prop:balance_increasing}
Let $\mathcal C$ be a currency that is good for the rich. Then, for $C'_1$ within the bounds specified by fairness in Definition~\ref{def:fairness}, 
\begin{align*}
    \val (C_1) \leq \val (C'_1)  
    \implies \BalanceII{V}{W}{C_1}\leq \BalanceII{V}{W}{C'_1} .
\end{align*}
\end{proposition}
\begin{proof}
If the currency is good for the rich, then if $\val (C_1) \leq \val (C'_1)$,   
anything that can be implemented with $C_1$ can also be implemented with $C'_1$ while inducing the same difference in value on the currency (within the bounds specified by Definition~\ref{def:fairness}). But then from the expression for the balance in Definition~\ref{def:balance},
\begin{align*}
\BalanceII{V}{W}{C_1}:= \sup_{C_2\in\mathcal{C}}\{\val C_2 -\val C_1: \
V \cap C_1 \to W\cap C_2\},
\end{align*}
it  follows that 
\begin{align*}
    \BalanceII{V}{W}{C_1} \leq \BalanceII{V}{W}{C'_1}.
\end{align*}
\end{proof}

\subsubsection{Having more does not help}

Here, we discuss the implications of a currency that is good for the poor. In such a theory, having more does not help | there are no discounts for the rich.

\begin{proposition}[Balance non-increasing  with available currency]
\label{prop:balance_decreasing}
Let $\mathcal C$ be a currency that is good for the poor. Then for $C'_1$ within the bounds specified by fairness in Definition~\ref{def:fairness},
\begin{align*}
    \val (C_1) \geq \val (C'_1 )
    \implies \BalanceII{V}{W}{C_1}\leq \BalanceII{V}{W}{C'_1} .
\end{align*}
\end{proposition}
\begin{proof}
If the currency is good for the poor, then if $\val (C_1) \geq \val (C'_1)$,   
anything that can be done with $C_1$ can also be implemented with $C'_1$ (within the bounds specified by Definition~\ref{def:fairness}). But then from the expression for the balance in Definition~\ref{def:balance}
it  follows that 
\begin{align*}
    \val (C_1) \geq \val (C'_1)  
    \implies 
    \BalanceII{V}{W}{C_1} \leq \BalanceII{V}{W}{C'_1}.
\end{align*}
\end{proof}

\subsubsection{Both directions}

As explained in Section~\ref{sec:nontechnical}, fairness guarantees that, within suitable bounds, whether or not a transformation can be facilitated with currency only depends on the difference in value between initial and final currency resource in the wallet, but not on their absolute values. This of course implies that also the balance of a resource transformation only depends on this difference and not on the available currency $C$ in the wallet, as stated in the following proposition.

\begin{restatable}[Fairness: balance is independent of available currency]{proposition}{propBalanceUnique}
\label{prop:balance_unique}
Let $(S^\Omega,\cT)$ be a resource theory equipped with a fair currency $\mathcal C$ for a target $\mathcal S$. Then, for any two resources $V,W\in\mathcal S$ and all $C,C'\in \mathcal C$ within suitable bounds (as specified by fairness),  
$$ \BalanceII{V}{W}{C} = \BalanceII{V}{W}{C'}, $$ 
and we can define 
$$ \Balance{V}{W} := \sup_{C\in\mathcal C} \BalanceII{V}{W}{C} $$
as the unique balance of the transformation from $V$ to $W$ in terms of the currency (up to boundary effects).
\end{restatable}
\begin{proof}
It follows directly from Propositions~\ref{prop:balance_increasing} and~\ref{prop:balance_decreasing} that for any two resources $V,W\in\mathcal S$ and all $C,C'\in\mathcal C$, within suitable bounds (as specified by fairness)
\begin{align*}
    \BalanceII{V}{W}{C} = \BalanceII{V}{W}{C'}.
\end{align*}
Then, since the balance of transforming $V$ into $W$ always exists for some initial $C$ (at worst, it corresponds to the cost of $W$), the quantity
$$ \Balance{V}{W}:=\sup_{C \in \mathcal C} \BalanceII{V}{W}{C} $$
is always well-defined.
\end{proof}

As discussed in Section~\ref{sec:nontechnical}, the cost and the yield of resources are then related to the balance. Firstly, we can express the cost and yield of resources in terms of the balance as follows.

\begin{restatable}[Cost and yield in terms of balance, for fair currencies]{theorem}{propCostYieldBalance}
\label{thm:cost_yield_balance}
Let $(S^\Omega,\cT)$ be a resource theory equipped with a fair currency $\mathcal C$ for a target $\mathcal S$. Then, for any resource $V\in\mathcal S$,
\begin{align*}
    \Cost{V} &= - \Balance{\Omega}{V}, \\
    \Yield{V} &= \Balance{V}{\Omega}.
\end{align*}
\end{restatable}

\begin{proof}
To show the first statement, note that $ \Omega \in \mathcal C $ and so for any $C\in\mathcal C$ and $V \in \mathcal S$,
\begin{align*}
    \Cost{V} &:= \inf_{C\in\mathcal C} (\val(C): C \to V) \\
    \flag{\text{fairness}}&= \inf_{C,C'\in\mathcal C} (\val(C) - \val(C'): C \to V \cap C') \\
    &= - \Balance{\Omega}{V}.
\end{align*}
For the second statement, note that $\Omega$ is the least valuable currency resource, and $\Yield{V}\geq 0$. This means that boundary conditions are not relevant and so fairness guarantees that
$$ \Yield{V} = \BalanceII{V}{\Omega}{\Omega} = \Balance{V}{\Omega}. $$
\end{proof}

Secondly, we can also relate the balance for an arbitrary resource transformation to the yield and cost of the initial and final resource respectively.

\begin{theorem}[Balance bounds in terms of cost and yield, for fair currencies]
\label{thm:balance_cost_yield}
Let there be a resource theory with a fair currency $\mathcal{C}$ for a target $\mathcal S$. Then for any $V,W\in \mathcal S$,
$$\Balance{V}{W}\geq \Yield{V}-\Cost{W}-\tilde \eps$$
for arbitrarily small or zero $\tilde \eps\geq 0$ (depending on whether or not the cost and yield are achievable). \\
If the currency is tight for resources $V,W$, furthermore
$$\Balance{V}{W}=\Yield{V}-\Cost{W}=\Cost{V}-\Cost{W} .$$
\end{theorem}

\begin{proof}
This proof is split as follows:
\begin{enumerate}
    \item General fair currency 
    \begin{enumerate}
        \item Case $\Cost{W}<\Yield{V}$
        \item     Case $\Cost{W}\geq \Yield{V}$ 
    \end{enumerate}
    \item Tight currency
    \begin{enumerate}
        \item Case $\Balance{V}{W} > 0$
        \item Case $\Balance{V}{W} \leq 0$
    \end{enumerate}
\end{enumerate}

Before we begin, note that since the cost is defined as $\Cost{W} = \inf_{C \in \mathcal C} \{\val(C): C \to W \}$,
there will be a currency element $C\in\mathcal C$ such that 
$\val (C )= \Cost{W} + \eps$ for an arbitrarily small $\eps \geq 0$ and $C \to W $ (because the cost may not be attainable). The same reasoning holds for the yield and balance. For fair currencies, these quantities are all attainable and as such $\eps=0$.

\begin{enumerate}
    \item  General fair currency.

    \begin{enumerate}
 
       \item Case $\Cost{W}<\Yield{V}$

In this case, there exist elements $C_y, C_c\in \mathcal C$ such that $V \to C_y \to C_c\to  W$ and
$$\Yield{V}\geq\val(C_y) > \val(C_c)  \geq \Cost{W},$$
with 
$ \val (C_y )= \Yield{V}-\eps $ and
$ \val (C_c) = \Cost{W} + \eps' $
for some arbitrarily small or zero $\eps,\eps'\geq 0$ (depending on whether or not $\Yield{V}$ and $\Cost{W}$ are achievable).

Then, by fairness, there exists another currency resource $C_\alpha$ such that $ V \to C_y \to W \cap C_\alpha $ and
\begin{align*}
   0< \val (C_\alpha )&= \val (C_y) - \val (C_c) \\
    &= \Yield{V} - \eps - (\Cost{W} + \eps') .
\end{align*}
We obtain for the balance
\begin{align*}
    \Balance{V}{W}
    &= \sup_{C,C''} \{\val (C) - \val (C''): C''\cap V \to C  \cap W \} \\
    &\geq \val (C_\alpha )
    = \Yield{V} -\Cost{W} - \underbrace{(\eps + \eps')}_{\tilde \eps} .
\end{align*}

        \item Case $\Cost{W}\geq \Yield{V}$
        
Let $C_c\in\mathcal C$ be an achievable cost value of $W$ with
$ C_c \to W$ 
and
$\val(C_c) = \Cost{W} + \eps$ for arbitrarily small or zero $\eps\geq 0$. Then, due to fairness,
\begin{align*}
    \Yield{V}
    &= \Balance{V}{\Omega} \\
    &= \sup_{C''} \{\val (C_c) - \val (C''): C''\cap V \to C_c  \cap \Omega \} \\
    &= \val (C_c) - \val (C') + \eps'
\end{align*}
for some $C'\in\mathcal C$ such that $ C'\cap V \to C_c $ and for an arbitrarily small or zero $\eps'$. 
But then we can implement 
$ V \cap C'\to C_c \to W $,
and so
$$ \Balance{V}{W} \geq -\val(C') = \Yield{V}- \Cost{W} - \eps - \eps'$$
for arbitrarily small or zero $\eps,\eps'\geq 0$ (depending on whether or not $\Yield{V}$ and $\Cost{W}$ are achievable).

    \end{enumerate}
    \item Currency is tight for $V$ and $W$.
    
    Note that in this case, the yield and cost are attainable.
    \begin{enumerate}
    
        \item Case $\Balance{V}{W} > 0$
        
In this case
\begin{align*}
    \Balance{V}{W} &= \sup_{C_1,C_2\in\mathcal C} (\val C_1 - \val C_2: V \cap C_2 \to W \cap C_1) \\
    \flag{\text{fairness, } \Balance{V}{W}> 0 }
    &= \sup_{C_1\in\mathcal C} (\val C_1: V \to W \cap C_1) \\
    \flag{\Cost{W}=\Yield{W}}
    &= \sup_{C\in\mathcal C} (\val C - \Cost{W}: V \to C) \\
    &= \Yield{V}-\Cost{W},
\end{align*}
where in the third line we used that for a fair currency that is tight for $W$,
$ W \cap C_1 \rightleftharpoons C $
when $\val C = \val C_1 + \Cost{W}$.

        \item Case $\Balance{V}{W} \leq 0$
        
  In this case,      
        \begin{align*}
    \Balance{V}{W} &= \sup_{C_1,C_2\in\mathcal C} (\val C_2 - \val C_1: V \cap C_1 \to W \cap C_2) \\
    \flag{\text{fairness, } \Balance{V}{W}\leq 0 }
    &= \sup_{C_1\in\mathcal C} (-\val C_1: V \cap C_1 \to W ) \\
    \flag{\Cost{V}=\Yield{V}}
    &= \sup_{C\in\mathcal C} (-\val C+\Yield{V}: C \to W ) \\
    &= \Yield{V}-\Cost{W},
\end{align*}
where in the third line we used that $ V \cap C_1 \rightleftharpoons C$ with $\val C = \val C_1 + \Yield{V} $.

    \end{enumerate}
\end{enumerate}

\end{proof}

\section{Examples}
\label{appendix:example}
In this appendix we apply our notions of currencies to three resource theories: LOCC, thermal operations and unital maps.

\subsection{Bell pairs in LOCC}

Here
show that the familiar example of Bell pairs in LOCC satisfies Definitions~\ref{def:currency} and~\ref{def:independence} and can be understood as a Stage II currency.

\begin{proposition}[Bell pairs are a Stage II currency]
\label{prop:Bell_currency}

Consider the resource theory of two-party LOCC over the partition $A'\otimes (\bigotimes_i^N  A_i )  \| B'\otimes (\bigotimes_i^N  B_i) $, in  a global Hilbert space 
$$\hilbert_{\text{global}} = \underbrace{\left( \bigotimes_i^N ( A_i  \otimes B_i) \right)}_{\text{wallet}} \otimes  \underbrace{A' \otimes \tilde B'}_{\text{target}},$$
where each $A_i$ and $B_i$ is a qubit, and $A'$, $B'$ are composed of $N$ qubits each.

Take the  currency $\mathcal C$ composed of specifications of a certain number $n$ of copies of Bell pairs in the wallet,
$\mathcal C = \{ {\Psi^n} \}_{n=1}^N \cup \{ \Psi^0=\Omega\},$
with 
${\Psi^n} = \{ \sigma_{\text{global}}: \sigma_{A_1 B_1 \dots A_n B_n} = \pure \psi^{\otimes n} \} ,$
and $\ket \psi = (\ket{00} + \ket{11} )/ \sqrt2$.

Take the target $\mathcal S$ to be any set of specifications that are local in $A'\otimes B'$, for example $\mathcal S = \{\local{\rho_{A'B'}} \}_{\rho}$, with
$$\local{\rho_{A'B'}}= \{ \sigma_{\text{global}}: \sigma_{A'B'} = \rho_{A'B'} \} .$$

Then $\mathcal C$ forms a currency for target $\mathcal S$ according to Definition~\ref{def:currency} and satisfies the independence condition~\ref{def:independence}.

\end{proposition}
\begin{proof}

For $\mathcal C$ to satisfy Definition~\ref{def:currency}, we need to check the following things:
\begin{enumerate}
\item Order: $\forall\, \Psi^n, \Psi^{m}\in \mathcal{C}$, either $\Psi^n\to \Psi^{m}$ or $\Psi^{m}\to \Psi^{n}$.

This is true since for $n\geq m$, $ \Psi^n \subseteq \Psi^m $, and so trivially by mere forgetting of information on some of the qubits
$$ n\geq m \implies \Psi^n \to \Psi^{m}. $$

Physically, any allowed transformation on  the wallet can only degrade the currency resource and result in a loss of available maximally entangled qubits.

\item Universality: for all target specifications $V\in \mathcal S$, there is a 
$\Psi^n\in \mathcal C$
such that 
$\Psi^n\to V$, and there is a $\Psi^m\in \mathcal C$
such that 
$V\to \Psi^m$. 

Due to teleportation, any bipartite state $\rho_{A'B'}$ on a certain number of qubits can be generated from a corresponding amount of maximally entangled qubits under LOCC~\cite{Bennett1982}. Going to a more general specification containing at least one local state is just a matter of forgetting. Hence, one can achieve any transformation
$$ \Psi^n\to V $$
for some $n\leq N$. Also, trivially $V \to \Omega$ and $\Omega\in\mathcal C$.
\end{enumerate}

To see that $\mathcal C$ also satisfies the independence condition in Definition~\ref{def:independence}, note that 
because the  wallet and the target $A'B'$ are different subsystems,
any state in the wallet is compatible with any state in the target 
(at the very least, we can combine any two local states with the tensor product). 
Furthermore, indeed $\Psi^n\to \Psi^{m}$ implies we can do so without disturbing the target: 
this is true 
since obviously merely forgetting information about some of the qubits in the wallet does not affect the state on the systems $A'B'$. 
More generally, any transformation in the wallet (such as one that replaces the maximally entangled state on some of the qubits by the state $\ket{00}$) on the two parties leaves the state on $A'B'$ invariant.

Finally, we can assign the value function $\val(\Psi^n)=n$ since $n\geq m$ is a necessary and sufficient condition for the transformation $\Psi^n\to \Psi^{m}$. 
Note that $n=0$ corresponds to the adopted convention of minimal value for $\Omega$.
\end{proof}

\subsection{Work in thermal operations}

Here we show that the usual notion of work in thermal operations (for block-diagonal states) is a Stage II currency.
This is an example of a currency that is not made of copies of individual states, that is, a currency without a subsystem structure inside the wallet. 

Just like in the case of Proposition~\ref{prop:Bell_currency}, this statement constitutes merely an example to show how for a specific target system we can construct a currency, and so it does not constitute a rigorous analysis of the most general cases of currencies in this theory.
For example, it would be interesting to study a currency made out of confidence regions around fixed values of energy \cite{Aberg2013a}.

\begin{proposition}[Stage II currency in quantum thermodynamics]
\label{prop:work_currency}
Consider the resource theory of thermal operations in a global Hilbert space $S \otimes W$, where each subsystem has dimension $d$, and identical individual Hamiltonians $H_i= \sum_k k \, \Delta \, \pure {E_k}_i$. 
Then the currency $ \mathcal{C} = \{C^k\}_k \cup \{\Omega \} $,
with 
$$ C^k = \{ \sigma_{CS}: \tr_S(\sigma_{CS}) = \pure{E_k}_C \} $$
is a Stage II currency for the target made out of local block-diagonal states on $S$, with the value function 
$\val(C^k) = E_k$,  $\val(\Omega)=0$.
\end{proposition}

\begin{proof}
For block-diagonal states in the energy eigenbasis, whether or not two states can be transformed into one another in the resource theory of thermal operations can be determined by means of the thermo-majorization condition, which amounts to comparing the corresponding Lorentz curves of the quantum states~\cite{HorodeckiOppenheim2013,Gour2013}.
We may check the following points:
\begin{enumerate}
\item Order: $\forall \, C,C' \in \mathcal{C}$, either $C\to C'$ or $C'\to C$.

This is true since for $ E_i \geq E_j $, thermo-majorization guarantees that 
$ C^i \to C^j $.
Also, trivially $ C  \to \Omega $ for any $C\in\mathcal C$.

\item Universality: for all target specifications $V,W\in \mathcal S$, there is a $C\in \mathcal{C}$ such that 
$C \to V$, and a $C'\in\mathcal{C}$ such that $V\to C'$.

It is clear that the pure state on the maximum energy eigenvalue on $S$ thermo-majorizes any other state on $S$ that is block-diagonal in the energy eigenbasis. Hence, starting from the largest energy eigenstate on system $C$ we can simply swap systems $C$ and $S$ and then generate the required state on $S$, thus achieving any transformation
$ C_\text{max} \to V $,
with 
$C_\text{max}=
\{\sigma_{CS}: \tr_S(\sigma_{CS}) = \ket{E_\text{max}}\bra{E_\text{max}}_C \} $ and $E_\text{max}$ the maximum energy eigenvalue.
Also, trivially $V \to \Omega$ and $\Omega\in\mathcal C$.

\item Non-disturbance: for $C,C'\in\mathcal C$, we have that $C\to C'$ implies we can do so without disturbing the target.

Transformations on system $C$ do not affect the states in the target system $S$ since local quantum maps are of the form $\mathcal{E}_C\otimes \id_S$.

\item Compatibility: all $C \in\mathcal{C}$ can be composed with any target resource.

This is true because $C$ and $S$ are different subsystems, and so any state on $C$ is compatible with any state on $S$ (at the very least, we can combine any two states on $C$ and $S$ with the tensor product).
\end{enumerate}

\end{proof}

\subsection{Resource theory of unital maps}

\subsubsection{Stage I currency}

We now give an explicit Stage I currency for the resource theory of unital maps.

\begin{restatable}[Currency for unital maps]{proposition}{propCurrencyUnitalMaps}
\label{prop:currency_unital}
The set
$ \mathcal{C} 
= \{ C^k \}_{k \in \{1,\dots d\}} \cup \{ \Omega \} $
with
$ C^k = \left\{ \frac{\Pi^{k}}{k} \right\} $, where $\Pi^k$ denotes the projector onto the first $k$ vectors in a given basis, forms a  universal currency under unital maps on a Hilbert space of dimension $d$.
\end{restatable}

\begin{proof}
Under unital maps, the pre-order on states is given by majorization \cite{Birkhoff1946, Hardy1952, Mendl2008, Gour2013, Weilenmann2015}. 
We need to check the following things:
\begin{enumerate}
    \item Order: the set $ \mathcal{C} $ of the currency is ordered (up to equivalence) by $\to$. 
    
    This holds because the states in $C$ are ordered by the rank (there is a clear order given by majorization coming from unital maps on $C$, see e.g.~\cite{Weilenmann2015}).
    
    \item Universality: for all target specifications $V$ (in this case all specifications on the whole $d$-dimensional system), there exists an element of the currency $ C \in \mathcal{C} $ such that
    $ C \to V $
    and an element of the currency 
    $ C' \in \mathcal{C} $ such that
    $ V \to C'$. 
    
    To see why the first statement holds, note that in order to prepare $V$, it is enough to prepare any one state $\rho\in V$, such as the state with the smallest maximum eigenvalue. 
    As majorization tells us (see e.g. Ref.~\cite{Weilenmann2015}), this will always be possible with unital maps from a currency state with small enough rank (in the worst case, from the pure state). 
    
    To show the second statement, note that $\Omega\in\mathcal C$, and we can always achieve $ V \to \Omega$.
\end{enumerate}
\end{proof}

In terms of this currency, we can now determine the cost of a specification as follows.

\begin{restatable}[Cost for unital maps]{proposition}{propCostUnitalMaps}
\label{prop:cost_majorization}
The cost of a specification $V$ under unital maps in terms of the currency $\mathcal{C}$ as above is given by
$$ \Cost{V} 
= \inf_{\rho \in V} \log m(\rho), $$
where 
$$ m(\rho) = 
    \frac{d}{\lfloor \lambda_\text{max}^{-1} (\rho) \rfloor} = \frac{d}{\lfloor 2^{H_\text{min}(\rho)}\rfloor}$$
with $\lambda_\text{max}(\rho) $ the largest eigenvalue of $\rho$ and $\lfloor \cdot \rfloor$ denoting the nearest integer (from below) to the enclosed expression.
\end{restatable}

\begin{proof}
This proof consists of two parts: first, we reduce the problem of finding the cost of a specification $V$ to an optimization over simple state transformations. Then we apply known results on majorization and follow the results in Ref.~\cite{Weilenmann2015} to solve it.

Since unital maps take quantum states to quantum states,
$$ \left\{ \frac{\Pi^{k}}{k} \right\} 
\to V $$
is true if and only if there is a state 
$\rho\in V$ such that
$\left\{ \frac{\Pi^{k}}{k} \right\} 
\to \{\rho\}$.
Then we can employ the majorization condition, which states that under unital maps
$$ \rho \to \sigma \iff 
\forall j \leq d, \ 
\sum_{i=1}^j \lambda_\rho^i \geq \sum_{i=1}^j \lambda_\sigma^i, $$
where $\lambda_\rho^i$ denotes the $i$-th eigenvalue of $\rho$ (where we number the eigenvalues in decreasing order). Since the currency state
$ \frac{\Pi^k}{k} $ is uniform, we can implement 
$$ \frac{\Pi^k}{k} \to \rho $$
with unital maps if and only if
$$ \frac{1}{k} \geq \lambda_\text{max} (\rho). $$
Hence we find that indeed
\begin{align*}
     \Cost{V} &= \inf_{\rho\in V,k\leq d}
    \left(\log d - \log k: \left\{ \frac{\Pi^k}{k} \right\} 
    \to \{\rho \} \right) \\
    &= \inf_{\rho\in V,k\leq d}
    \left(\log d - \log k: \frac{1}{k} \geq \lambda_\text{max}(\rho) \right) \\
    &= \inf_{\rho\in V}
    \left(\log d - \log \lfloor \lambda_\text{max}^{-1} (\rho) \rfloor \right) .
\end{align*} 
\end{proof}

Similarly, we find the following result for the yield of a specification.

\begin{restatable}[Yield for unital maps]{proposition}{propYieldUnitalMaps}
\label{prop:yield_majorization}
The yield of a specification $V$ under unital maps in terms of the currency $\mathcal{C}$ as above satisfies
\begin{align*} 
\Yield{V} 
&=
\log d - \max_{\rho \in V^\P}  H_0 (\rho),
\end{align*}
where $V^\P$ denotes the convex hull of $V$~\cite{DelRio2015}, and $H_0(\rho) = \log \rank (\rho)$.
\end{restatable}

\begin{proof}
First we show that 
$
\max_{\rho \in V^\P}  H_0 (\rho) = 
\log d_{\text{eff}} ,
$
where $d_{\text{eff}} \leq d$ is the smallest number such that there exists a basis $\{\ket i\}_{i=1}^d$ such that 
$\tr(\rho\, \Pi^{d_{\text{eff}}}) =1$ for all $\rho \in V$, with 
$$ \Pi^{d_{\text{eff}}} = \sum_{i=1}^{d_{\text{eff}}} \pure{i}.$$
In other words, $d_{\text{eff}}$ is defined such that each element $\rho \in V$ can be expressed in the basis  $\{\ket i\}_{i=1}^d$  as 
\begin{align*}
    \rho= \begin{pmatrix}
    \begin{array}{c c c|}
         & & \\
         & \rho_{\text{eff}}& \\
         & &\\
         \hline
    \end{array}
    & 0 &\cdots &0\\
    0  & 0 \\
    \vdots & &\ddots \\
      0  && &0 
    \end{pmatrix},
\end{align*}
where $\rho_{\text{eff}}$ acts on a Hilbert space of dimension $d_{\text{eff}}$.
In particular, this means that we can express all elements of $V$ in the $d_{\text{eff}}-$dimensional basis $\{\ket i\}_{i=1}^{d_{\text{eff}}}$. 
This implies that $V$ itself lives in a vector space $\Gamma$ of dimension  $d_{\text{eff}}^2-1$, because it is composed of density matrices that can be expressed using only $d_{\text{eff}}$ basis elements.

Now we can look at the cone of positive semi-definite matrices on Hilbert spaces of dimension $d_{\text{eff}}$ (of which normalized density matrices form an affine slice). The extremal faces of this cone are formed by  
low-rank matrices, and the interior of the cone is made of all full-rank matrices,  \cite[Chapter II, Proposition 12.3]{Barvinok2002}.  
The density matrices in the set $V$ cannot all lie on the same face of the cone, as this would contradict the definition of $ d_{\text{eff}}$ (because then we would be able to project all points of  $V$ onto the same lower-dimensional space). 
This implies that if we take the convex hull of $V$, some points will lie in the interior of the cone and therefore have full rank, which implies 
$$\exists \, \rho \in V^\P: \quad H_0 (\rho) = \log d_{\text{eff}}.$$ 
This is also the maximal entropy  for $d_{\text{eff}}-$dimensional density matrices, and so we have shown that  
$\max_{\rho \in V^\P} H_0 (\rho) = \log d_{\text{eff}}$.

To show that $\Yield{V} \geq  \log d - 
\log d_{\text{eff}}$, we apply the unital map 
$$\E(\rho) = \frac{\tr(\rho \, \Pi^{d_{\text{eff}}})}{ d_{\text{eff}}} \  \Pi^{d_{\text{eff}}} + (\id_d -  \Pi^{d_{\text{eff}}} ) \, \rho \, (\id_d - \Pi^{d_{\text{eff}}} )$$
to each element $\rho \in V$. 
Because of the way we defined $\Pi^{d_{\text{eff}}} $, we obtain 
$$ \E (V) = \frac{\Pi^{d_{\text{eff}}} }{d_{\text{eff}}} \in C^{d_{\text{eff}}} ,$$
which is a currency element, and therefore 
$$ \Yield{V} \geq \val(C^{d_{\text{eff}}}) = \log d - \log d_{\text{eff}}.$$

To show that 
$\Yield{V} \leq \log d - \log d_{\text{eff}}$,
suppose that there exists a number $k < d_{\text{eff}}$ such that $V\to C^k$. That is, there is some unital map $\E'$ that maps every element in $V$ to the state $\Pi^k/k$. 
In particular, $\E'$ maps all elements of $V$ to the same face of the cone. 
But then, by linearity $\E'$ also maps convex combinations of states in $V$, which as we showed before lie in the interior of the cone (unless $V$ was already on a face, which would contradict the definition of $d_\text{eff}$), to $\Pi^k/k$. This is however not possible, because the rank of a state cannot decrease under unital maps (due to the majorization condition), and states in the interior have full rank  $d_\text{eff}$ while $\text{rank}(\Pi^k/k)=k<d_\text{eff}$ by assumption.  
Hence we have reached a contradiction.

\end{proof}

\subsubsection{Stage II currency}

Assuming a bipartite system $W\otimes S$, we can construct a Stage II currency for resources on system $S$ which consists of local resources on $W$.

\begin{restatable}[Stage II currency for unital maps]{proposition}{propCurrencyUnitalMapsII}
\label{prop:currency_unitalII}
Consider the resource theory of unital maps in a global Hilbert space $W \otimes S$, with $d_W := \dim W \geq \dim S =: d_S$.
The set
$ \mathcal{C} 
= \{ C^k \}_{k \in \{1,\dots d_W\}} \cup \{ \Omega \} $,
with
$$ C^k = \left\{ \sigma_{WS}: \tr_S(\sigma_{WS}) = \frac{\Pi^k_W}{k} \right \} ,$$
where $\Pi^k_W$ denotes the projector onto $k$ states in a given basis on system $W$,  
forms a currency for the target defined through
$$ V \in \mathcal S \iff V = \{ \sigma_{WS}: \tr_W(\sigma_{WS}) \in V_S \} $$
for some set
$ V_S $ of  reduced density matrices on system $S$. This currency furthermore satisfies \emph{independence} according to Definition~\ref{def:independence}.
\end{restatable}
\begin{proof}
We need to check the following things:
\begin{enumerate}
    \item Order: the set $ \mathcal{C} $ of the currency is ordered (up to equivalence) by $\to$.
    
    For the specification $C^k$, it holds that
    $$ C^k \leftrightharpoons  \{\tilde \rho_{SW}\}=\left\{\frac{\Pi^k_W}{k} \otimes \frac{\id_S}{d_S} \right\} $$
    since $\tilde \rho_{SW} \in C^k$ and we can always use unital operations to replace the state on $S$ by the fully mixed state.
    But then, the majorization condition for  unital maps on system $W$ tells us that 
    $$ k'\geq k \quad \iff \quad  \frac{\Pi^k_W}{k} \otimes \frac{\id_S}{d_S} \to \frac{\Pi^{k'}_W}{k'} \otimes \frac{\id_S}{d_S}. $$
    Since these states are interchangeable with the respective currencies, we have that 
    $ C^k \to C^{k'} $
    if and only if $k'\geq k$. 
    Finally, 
    $ C^k \to \Omega $
    for any $k$, and also 
    $ \Omega \to C^{d_W} $, 
    since we can always replace the state in $W$ by a fully mixed state by means of a unital operation. 
    
    Hence the currency resources are ordered (up to the equivalence of $C^{d_W}$ to $\Omega$). 
    
    \item Universality: for all target specifications $V \in \mathcal S$, there exists an element of the currency $ C \in \mathcal{C} $ such that
    $ C \to V $,
    and an element $C'\in\mathcal C$ such that
    $ V \to C'.$
    
    To see why the former holds, note that given $ d_W = \text{dim} (W) \geq \text{dim} (S) = d_S $, 
    we can always start from 
    the currency resource 
    $C^1 = \{ \sigma_{SW}: \tr_S(\sigma_{SW}) = \ket{1}\bra{1}_W \} $ and 
    apply the unitary
    $$ U_{SW} = \sum_{i\neq j;i,j=1}^{d_S} \ket{ij}\bra{ji}_{SW} + \sum_{i=d_S+1}^{d_W} \ket{ii}\bra{ii}_{SW}, $$
    where the second sum vanishes if $d_S=d_W$. This unitary essentially achieves a 
    swap of the state on system $S$ with part of $W$ of the same size. In particular, the reduced state on $S$ is now $\ket{1}\bra{1}_S$, and so one can reach the specification
    $$ Y= \{ \sigma_{SW}: \tr_W(\sigma_{SW}) = \ket{1}\bra{1} \} .$$
    We now use the idea from the proof of order, namely that $Y$ is interconvertible with the following state,
    $$ Y \leftrightharpoons \left\{ \tau_{SW}\right\}=\left\{\ket{1}\bra{1}_S \otimes \frac{\id_W}{d_W} \right\}. $$
    Then we can prepare any state in $V_S$ in system $S$ from the pure state in $S$ by means of a unital map on $S$ (since the pure state majorizes every other state, this is possible), that is, for any $\rho_S\in V_S$,
    $$ \left\{ \tau_{SW} \right\} \to \left\{\rho_S \otimes \frac{\id_W}{d_W} \right\}. $$
    As a result, one can reach a specification
    $$ X = \{ \sigma_{SW}: \tr_W(\sigma_{SW}) = \rho_S \} $$
    for any particular $ \rho_S \in V_S$. 
    But then clearly
    $ C_1 \to V. $
    
    To see why the second statement holds, note that $\Omega\in \mathcal C$ and clearly $V\to\Omega$ for any $V\in\mathcal S$.

    \item Compatibility: all specifications $ C^k \in \mathcal{C} $ (including $\Omega$) can be composed with any specification $V$ in the target $ \mathcal S $.
    
    This holds since we can trivially compose any state $\rho_S\in V_S$ in $S$ with any state $\frac{
    \Pi^k_W}{k} $ in $W$ via the tensor product. 
    Since both $V$ and $C^k$ contain only local information about the systems $S$ and $W$ respectively, at the very least  $\frac{
    \Pi^k_W}{k}  \otimes \rho_S \in V\cap C^k $.
    
    \item Non-disturbance: for two currency resources $C^k,C^l\in\mathcal C$, if we can transform $C^k \to C^l$, then we can do this without disturbing the target, i.e.\ 
    $$ C^k \to C^l \implies V \cap C^k \to V \cap C^l .$$
    
    This holds because a unital map on $W$ used in order to convert between the currency states only involves system $W$ and leaves the target space unchanged.
\end{enumerate}

\end{proof}

We may now compute the cost and yield of specifications.

\begin{theorem}[Cost and yield for unital maps with Stage II currency]
\label{thm:cost_yield_unitalII}
Let $S\otimes W$ be a bipartite system with
$d_W \geq d_S$. 
Then, in terms of the currency
$\mathcal C$ of Proposition~\ref{prop:currency_unitalII},
the cost and yield of  resources $V\in\mathcal S$
satisfy 
$$ \inf_{\rho_S \in V_S} \log d_S -  H_\text{min}(\rho_S) \leq \Cost{V} 
\leq \inf_{\rho_S \in V_S} \log \left(\frac{d_S}{\lfloor 2^{H_\text{min}(\rho_S)}\rfloor}\right) $$
and 
\begin{align*} 
\Yield{V}
&= \log d_S - \max_{\rho_S \in V_S^\P}  H_0 (\rho_S).
\end{align*}
Furthermore, in the limit of large currency dimension $d_W$, the cost of a resource $V$ 
can approximate the value  
$ \inf_{\rho_S\in V_S} \log d_S - \hmin(\rho_S) $
arbitrarily closely.  
That is, for any $\eps > 0 $ and any target dimension $d_S$, we can find a currency dimension $d_W\geq d_S$ such that
$$ \inf_{\rho_S\in V_S} \log d_S - \hmin(\rho_S) \geq \Cost{V} \geq \inf_{\rho_S\in V_S} \log d_S - \hmin(\rho_S) - \eps .$$
\end{theorem}

\begin{proof}
We start from the definitions of cost and yield for this currency, which can be expressed as 
\begin{align*}
  \Cost{V} 
    &= \inf_{C\in\mathcal C} \{\val C: C \to V\} 
    = \log d_W - \sup_{1 \leq k \leq d_W} \{\log k: C^k \to V  \} \\
  \Yield{V} 
    &= \sup_{C \in\mathcal C} \{\val C: V \to C\} 
    = \log d_W - \inf_{1 \leq k \leq d_W} \{\log k: V \to C^k  \},
\end{align*}
for integer values of $k$.
We note that our  currency elements are inter-convertible with simpler ones, 
$$C^k = \left\{ \sigma_{WS}: \tr_S(\sigma_{WS}) = \frac{\Pi^{k}_W}{k} \right \}  
\rightleftharpoons 
\left\{ \frac{\Pi^{k}_W}{k} \otimes \frac{\id_S}{d_S} \right\}
=:  \tilde C^k,  
$$
because under unital maps we are always allowed to replace the state on $S$ by the maximally mixed state that is uncorrelated with $W$, and in the other direction  $\tilde C^k \subseteq C^k$.
Similarly, $V$ is inter-convertible with
$$ \tilde V = \left\{ \rho_S \otimes \frac{\id_W}{d_W} : \rho_S \in V_S \right\} $$
as we can always replace the state in $W$ by the maximally mixed state and in the other direction again $\tilde V \subseteq V$.
The definitions of cost and yield become
\begin{align*}
\Cost{V} &= \log d_W -  \sup_{k\leq d_W} \left\{ \log k : 
\left\{ \frac{\id_S}{d_S}\otimes \frac{\Pi^k_W}{k} \right\} 
\to \left\{  \rho_S \otimes \frac{\id_W}{d_W} : \rho_S \in V_S \right\} \right\} \\
\Yield{V} &= \log d_W -  \inf_{k\leq d_W} \left\{ \log k :
\left\{  \rho_S \otimes \frac{\id_W}{d_W} : \rho_S \in V_S \right\}
\to \left\{ \frac{\id_S}{d_S}\otimes \frac{\Pi^k_W}{k} \right\} \right\}.     
\end{align*}

For the cost, we can now argue analogously to Proposition~\ref{prop:cost_majorization}: since unital maps take states to states, the cost above is equivalent to the expression
$$ \Cost{V} = \log d_W -  \sup_{k\leq d_W, \rho_S \in V_S} \left\{\log k : 
\left\{ \frac{\id_S}{d_S}\otimes \frac{\Pi^k_W}{k} \right\} 
\to \left\{ \rho_S \otimes \frac{\id_W}{d_W} \right\} \right\}) .$$
Majorization  gives us a necessary and sufficient condition for this state transformation,
$$ \frac{\id_S}{d_S}\otimes \frac{\Pi^k_W}{k} 
\to \rho_S \otimes \frac{\id_W}{d_W}
\quad \iff \quad 
\frac{1}{k\  d_S} 
\geq  \frac{\lambda_\text{max}(\rho_S)}{d_W} $$
and so for the optimal $k$, 
$$ k = \sup_{\rho_S \in V_S} \left\lfloor \frac{d_W}{d_S \ \lambda_\text{max}(\rho_S)} \right\rfloor \geq \frac{d_W}{d_S} \sup_{\rho_S \in V_S} \left\lfloor \lambda_\text{max}^{-1}(\rho_S) \right\rfloor .$$
Hence 
\begin{align*}
 \Cost{V} 
   &= \log d_W - \log k \ \leq \ \log d_S - \sup_{\rho_S\in V_S} \log \lfloor \lambda_\text{max}^{-1}(\rho) \rfloor \\
   &= \inf_{\rho_S\in V_S} \log d_S - \log \left\lfloor 2^{\hmin(\rho_S)}\right\rfloor, 
\end{align*}
while in the other direction
\begin{align*}
  \Cost{V} 
    &\geq \log d_S + \inf_{\rho_S\in V_S} \log \lambda_\text{max}(\rho) \\
    &=  \log d_S - \sup_{\rho_S\in V_S} \hmin(\rho_S).
\end{align*}
Finally, we note that in the limit of large $d_W$, the real number $ d_S \cdot \lambda_\text{max}(\rho_S) \geq 1 $ can be approximated arbitrarily closely by a rational $\frac{d_W}{k}$ for $d_W,k\in\mathbb N$ with $k\leq d_W$. 
In particular, for any $\eps>0$, we can choose $d_W,k\in\mathbb N$ such that for the optimal $\rho_S$ 
$$ \log d_S - \hmin(\rho_S)\geq \log d_W - \log k \geq \log d_S - \hmin(\rho_S) - \eps .$$
But this implies that indeed
$$ \inf_{\rho_S\in V_S} \log d_S - \hmin(\rho_S) \geq \Cost{V} \geq  \inf_{\rho_S\in V_S} \log d_S - \hmin(\rho_S) - \eps. $$

For the yield, we can employ the same technique as in Proposition~\ref{prop:yield_majorization} because the output of the transformation
$$ \left\{  \rho_S \otimes \frac{\id_W}{d_W} : \rho_S \in V_S \right\}
\to \left\{ \frac{\id_S}{d_S}\otimes \frac{\Pi^k_W}{k} \right\} $$
from the expression for the yield above is a uniform state of rank $d_S \cdot k$ (we could see it as a currency state in $d=d_S\cdot d_W$). 
Hence, with $\tilde V = \left\{  \rho_S \otimes \frac{\id_W}{d_W} : \rho_S \in V_S \right\}$, we find
$$ \Yield{V} = \log d_S + \log d_W - \log d_\text{eff}(\tilde V) .$$
Then we note that 
$ d_\text{eff}(\tilde V) = d_W \cdot d_\text{eff}(V_S) $ and, again from Proposition~\ref{prop:yield_majorization}, 
$\log d_\text{eff}(V_S) = \max_{\rho_S\in V_S} H_0(\rho_S)$, 
and so
$$ \Yield{V} = \log d_S - \max_{\rho_S\in V_S} H_0(\rho_S) $$
as required.

\end{proof}

\subsubsection{Alternative currencies}

We may discuss alternative currencies for the resource theory of unital maps.

\begin{restatable}[Equivalent currency for unital maps]{proposition}{propEquivalentCurrency}
\label{prop:equivalent_currency}
The currency $\mathcal{C}$ of uniform states of different ranks is equivalent to the currency $\mathcal{C}_2  = \{C_2^k\}_k \cup \{\Omega\} $ made out of  the convex hull of differently sized sets of orthogonal pure states,
$$ C_2^k = \left(\{ \ket{i}\bra{i} \}_{1\leq i \leq k} \right)^\P , $$
in the sense that
$$ \forall \, C^k \in \mathcal{C}, C_2^k \in \mathcal{C}_2,  \quad
C^k \to C_2^k \text{ and } C_2^k \to C^k. $$

\end{restatable}
\begin{proof}
For the first statement, 
we define 
$$ \val(C_2) := \min_k \left(\log d - \log k: \frac{\Pi^k}{k}\in C_2\right) $$
and so it is easy to see that for 
$ C \in \mathcal{C} $,
$$ \val(C) = \val(C_2) \implies C \to C_2 $$
since $C \subseteq C_2$. 
The other direction, $ C_2 \to C $, can be easily established from a unital map that replaces the currency state by a uniform state of rank of the effective dimension of $C_2$ as in the proof of Proposition~\ref{prop:yield_majorization}, yielding precisely $C$. 
\end{proof}

\begin{proposition}[Nonequivalent currency]

The statement of Proposition \ref{prop:equivalent_currency} does not hold for a set $ \mathcal C_3 = \{ C_3^k \}_k \cup \{\Omega\} $ made of differently sized sets of orthogonal states without the convex hull,
$$ C_3^k = \{ \ket{i}\bra{i} \}_{1\leq i\leq k} ,$$
since (as a counterexample)
$$ \left\{ \frac{\id_2}{2}\right\} \nrightarrow \{\ket{1}\bra{1},\ket{2}\bra{2}\} ,$$
i.e.\ $C^2 \nrightarrow C_3^2$. 
The set $\mathcal C_3$ does nevertheless form a currency for $S^\Omega$. 
\end{proposition}

\begin{proof}
To see that it is indeed true that in a qubit $S$
$$ C^2 = \left\{ \frac{\id_2}{2}\right\} \nrightarrow \{\ket{1}\bra{1},\ket{2}\bra{2}\} = C_2^2,$$
it is enough to employ Proposition~\ref{prop:cost_majorization}: we see that in terms of the currency $\mathcal{C}$,
\begin{align*} 
&\Cost{\{\ket{1}\bra{1},\ket{2}\bra{2}\}} \\
&= \inf_{\rho \in \{\ket{1}\bra{1},\ket{2}\bra{2}\}} 
[\log 2 - \log \lfloor 2^{H_\text{min} (\rho)} \rfloor ] \\
&= 1 > 0 
= \Cost{\left\{ \frac{\id_2}{2} \right\}},
\end{align*}
and so since $\Cost{\cdot}$ is a monotone along the pre-order, 
$ \left\{ \frac{\id_2}{2} \right\} 
\nrightarrow \{\ket{1}\bra{1},\ket{2}\bra{2}\} $.

To see that $\mathcal C_3$ is still a currency, note that it satisfies order and universality for $S^\Omega$: universality holds because under unital maps, $ C_3^1 = \{\ket{1}\bra{1}\}$ can be transformed into any state $\rho$ (this follows from simple majorization), and so for any $V\in S^\Omega$,
$ C_3^1 \to V $. In the other direction, $V\to \Omega$ for any $V \in S^\Omega$ and $\Omega\in\mathcal C_3$. \\
To see why order holds, note that for $k'\geq k$, simply $ C_3^k \subseteq C_3^{k'} $ and so $C_3^k \to C_3^{k'}$. To see that this only holds in one direction, that is, that 
$C_3^{k'} \nrightarrow C_3^k$ for $k'> k$, note that if there was a map that takes $C_3^{k'}$ to $C_3^k$, then it would have to take all states in $C_3^{k'}$ into some state in $C_3^k$. But such a map would then also take the uniform state over the first $k'$ eigenvalues to something that has support only on the first $k$ eigenvalues, so that 
$$ \pi^{k'} := \left\{ \frac{\Pi^{k'}}{k'} \right\} \to \{\tau\} $$
for some $\tau$ such that $\tr(\Pi^k \tau)=1$. But then $H_0 (\tau) \leq H_0(\pi^{k'})$, which contradicts the fact that $H_0$ is a monotone under unital maps (this follows from the majorization condition). Hence 
$$ k'\geq k \iff C_3^k \to C_3^{k'} .$$
\end{proof}

\subsubsection{Approximations and probability of failure}

As a special case of our framework, we can model cases where we allow for some error probability in the output of a transformation, or where we need to make sure that our protocol is robust against errors in the initial resources. In the first case, the transformation will become cheaper to implement since we can save some currency by betting on a successful outcome. In the second case, the error makes it harder to still achieve the final resource, so that the transformation becomes more expensive.

To illustrate how this would work, we can look at the special case of resources that correspond to approximate quantum states, $\mathcal B^\eps(\rho)$, according to some metric like the trace distance or the purified distance based on the fidelity. More generally, we may look at cases beyond such metrics, namely arbitrary coarse-grainings of resources represented by so-called \emph{approximation structures} as introduced in~\cite{DelRio2015}, which again give us a notion of epsilon balls, $\mathcal B^\eps(\cdot)$, for parameters $0\leq \eps \leq 1$.

We can now introduce such \emph{smoothing} both at the input and the output of transformations, to address the two types of error probability. 
For example, we then find the following result for the cost of an approximate quantum state.

\begin{proposition}[Cost of approximations for unital maps]
In terms of the Stage I currency $\mathcal C$ introduced before, the  specification $\mathcal B^\eps(\rho)$ for some $0\leq \eps \leq 1$ (where the smoothing is taken according to some metric on state space such as the trace distance or the fidelity) satisfies
\begin{align*}
     \Cost{B^\eps(\rho)} &= \log d - \sup_{\sigma \in B^\eps(\rho)} \log \lfloor 2^{\hmin(\sigma)} \rfloor \\
     &= \log d - \log \lfloor 2^{\hmin^\eps(\rho)}\rfloor
\end{align*}
\end{proposition}
\begin{proof}
The proof of this follows directly from the definition of $\hmin^\eps$ as
$$ \hmin^\eps(\rho) = \sup_{\sigma\in\mathcal B^\eps(\rho)} \hmin(\sigma) ,$$
since then we immediately get 
\begin{align*}
    \Cost{B^\eps(\rho)} &= \log d - \sup_{\sigma \in \mathcal B^\eps(\rho)} \log \lfloor 2^{\hmin(\sigma)} \rfloor \\
    &= \log d - \log \lfloor 2^{\sup_{\sigma\in\mathcal B^\eps(\rho)} \hmin(\sigma)} \rfloor \\
    &= \log d - \log \lfloor 2^{\hmin^\eps(\rho)}\rfloor.
\end{align*}
\end{proof}

More generally, the larger the smoothing we introduce, the smaller the cost (Section~\ref{sec:technical}),
$$ \eps \geq \eps' \implies \Cost{\mathcal B^\eps(\rho)} \leq \Cost{\mathcal B^{\eps'}(\rho)} .$$

On the other hand, introducing the smoothing on the input, we can look at the yield of an approximate quantum state. 
Let us look at the particular example where we are interested in the yield of $\mathcal B^\eps(\ket{1}\bra{1})$ in a $d$-dimensional system for $0 < \eps < 1$. In this case, the resource in principle has full support on the whole basis, since for example 
$ \sigma = (1-\eps) \ket{1}\bra{1} +  \frac{\eps}{d-1} \left( \sum_{i=2}^{d} \ket{i}\bra{i} \right) \in B^\eps(\ket{1}\bra{1}) $ and $\Yield{\sigma}=0$. But then clearly the only resource (except $\Omega$) in the currency $\mathcal C$ that can be reached from $B^\eps(\ket{1}\bra{1})$ is $C^d=\{\frac{\id_d}{d}\}$, and so 
$ \Yield{B^\eps(\ket{1}\bra{1})}=0 $. 

More generally in the resource theory of unital maps, due to the same reasoning, the yield of an approximate state as specified by an error parameter $\eps>0$ without further limitations (such as additional knowledge about the state) or further relaxations (such as allowing to reach an approximate currency resource) will be zero for this currency $\mathcal C$. 
This is partly because the currency is relatively course-grained (it contains only flat states), but also because it consists of exact states rather than itself allowing for a probability of error. By adapting the currency in this way, we could hence arrive at different, perhaps more useful, statements about the yield of approximate states.

Naturally, we find that more generally for arbitrary approximation structures~\cite{DelRio2015} and currencies with large enough target to include the approximate states respectively,
$$ \eps \geq \eps' \implies \Yield{\mathcal B^\eps(\rho)} \leq \Yield{\mathcal B^{\eps'}(\rho)}.$$

More generally, we could look at the balance of going between smooth input and output specifications, for example
\begin{align*}
    V^{\eps} &:= \{ \sigma_{WS}: \sigma_S \in \mathcal B^\eps(\rho_S) \} ,\\
    W^{\eps'} &:=\{ \sigma_{WS}: \sigma_S \in \mathcal B^{\eps'}(\tau_S) \}.
\end{align*}
One immediate statement is that for $\delta \geq 0$,
\begin{align*} 
\BalanceII{V^{\eps}}{W^{\eps'}}C &\geq \BalanceII{V^{\eps+\delta}}{W^{\eps'}}C \\
\BalanceII{V^{\eps}}{W^{\eps'}}C &\leq \BalanceII{V^{\eps}}{W^{\eps'+\delta}}C .
\end{align*}
One may use other properties of the theory (like robustness under an approximation structure \cite{DelRio2015}) to tighten these bounds further.

\section{Pathological cases}
\label{appendix:pathological}
\subsection{Main ideas and statements}

Here we aim to get a better understanding of pathological resources, that is resources for which $\BalanceII{V}{V}{C}>0$ for some $C\in\mathcal C$, and see why such resources cannot exist if the currency is fair. 
To this end, 
it is instructive to weaken the definition of currencies by dropping monotonicity of value, that is, to allow for two currency resources $C_1$ and $C_2$ that $C_1 \to C_2$ even though $\val C_1 < \val C_2$. 
This allows us to study more precisely how pathological resources can help us generate more from less currency and when we may use them | under such a relaxed currency, we may also get resources $V\in\mathcal S$ for which $\Cost{V}<\Yield{V}$. 
In the end, note that we can recover the full order (up to equivalences) on the currency by ``renormalizing'' the value function, that is, by identifying all the values of inter-convertible currency resources, thus guaranteeing monotonicity of value and eliminating resources with $\Cost{V}<\Yield{V}$.\footnote{In fact, monotonicity of value is not essential for most of the proofs in this paper. We could also start from a weaker  (if less intuitive) definition of currency without a monotonous value, and check when that becomes relevant.}

We shall now look at the impact of fairness in one or both directions on resources $V$ with 
$$ \BalanceII{V}{V}{C}>0 $$
for some $C \in\mathcal C$, as well as their relation to resources for which $\Cost{V}< \Yield{V}$. We will also discuss what happens if we put back in monotonicity of value, and when this leads to a contradiction, that is, when the theory excludes pathological resources with $\BalanceII{V}{V}{C}>0$.

\subsubsection{Fairness in one direction}

\paragraph*{Good for the rich.} In a Stage II currency that is good for the rich,
we will see that if there are resources for which $\BalanceII{V}{V}{C}>0$, then the theory allows the agent to generate arbitrary amounts of currency as soon as the agent has access to $V$ and $C$.
More precisely, 
we can show that
if 
$\BalanceII{V}{V}{C} >0 $ for some $V\in\mathcal S$ and $C\in\mathcal C$, 
the theory essentially becomes trivial after a threshold of having access to $V$ and $C$, in the sense that (Proposition~\ref{prop:pathological_I}) 
$$ \BalanceII{V}{V}{C} \approx c_{\sup} - \val C .$$
This also means that if we would put back in monotonicity of value at this point (as long as the currency is not also fair in the other direction, we may do this), we would get that the cost of such a resource $V$ together with $C$ would essentially have to reach the top of the currency (Proposition~\ref{prop:cost_top}),
$$ \Cost{V} \gtrsim c_{\sup} - \val C .$$
This is because when one has enough currency to buy the resource $V$ and still keep the currency $C$ on the side, then one could generate arbitrary amounts of currency, so the theory needs to make sure that $ \Cost{V} + \val(C) $ essentially already corresponds to the highest possible value in the currency.

Finally, we can also look at how this relates to resources for which $\Cost{V}<\Yield{V}$ (if we were to allow them). 
We can show that if the theory is good for the rich,  then
\begin{align*} 
&\Cost{V}<\Yield{V} 
\implies \BalanceII{V}{V}{C}> 0
\end{align*}
for all currency resources $C\in\mathcal C$ (Proposition~\ref{prop:condition_balance_positive}). In particular 
$ \BalanceII{V}{V}{\Omega} > 0 $, which would imply that
$$ \Yield{V}\approx c_{\sup} .$$

\paragraph*{Good for the poor.} When the currency is good for the poor, we can show that a resource for which 
$\BalanceII{V}{V}{C}>0$ for some $C\in\mathcal C$ must also satisfy 
$\Yield{V}>\Cost{V}$. 
Namely, we get from $\BalanceII{V}{V}{C}>0$ that (Proposition~\ref{prop:pathological_II})
$$ \Yield{V}\geq\BalanceII{V}{V}{C}>0=\Cost{V}. $$
This of course would directly contradict the order on the currency as shown in Proposition~\ref{prop:cost_bigger_yield}, and so we may not reintroduce monotonicity here without obtaining a contradiction (see below).

\subsubsection{Fairness in both directions}
For a currency that is fair in both directions, 
we can show that (Corollary~\ref{cor:Midas_I})
\begin{align*}
    &\Cost{V} < \Yield{V} 
    \iff \Balance{V}{V}>0 .
\end{align*}
If such a resource exists, we furthermore see that the theory becomes essentially trivial, because 
\begin{align*} 
\Cost{V} = 0 \quad \text{ and } \quad
\Yield{V} &= \Balance{V}{V} = c_{\sup} 
\end{align*}
for any target resource $V\in\mathcal S$ (Corollary~\ref{cor:Midas_II}). Since such a theory allows us to generate arbitrary money for free like king Midas (because in particular $\Balance{\Omega}{\Omega}>0$), we call such a pathological theory a \emph{Midas theory} and the resource $V$ a Midas resource.

As fair theories are in particular good for the poor, it is easy to see that we obtain a contradiction when introducing back the monotonicity on value into the currency. Namely, all currency resources are inter-convertible, and we would have to set the value of all to be equal | but then we could not get $\Cost{V}>\Yield{V}$ or $\Balance{V}{V}>0$ in the first place.

\subsection{Formal statements and proofs}

\subsubsection{Having less does not help (good for the rich)}

\begin{proposition}[Pathological case I]
\label{prop:pathological_I}
Let $\mathcal C$ be a currency that is fair for the rich. Then, if the theory has a pathological resource $V\in\mathcal S$ such that 
$\BalanceII{V}{V}{C}>0$ for some $C\in\mathcal C$,
then
$$ c_{\sup} - \val (C) \geq \BalanceII{V}{V}{C}\geq c_{\sup} - \val (C) - \alpha $$
for some margin $\alpha$ that 
corresponds to the difference between $C$ and the next higher currency resource $C'$ \footnote{if there is no one next higher resource, one can choose a resource arbitrarily close to $C$}, $\alpha=\val(C') - \val(C)$.
\end{proposition}
\begin{proof}
Let $\BalanceII{V}{V}{C}>0$. Then for the next higher resource $C'\in\mathcal C$ \footnote{if there is no one next higher resource, we may choose a resource arbitrarily close to $C$} with 
$$ \alpha := \val(C') - \val(C) > 0 $$ it holds that
$$ V\cap C \to V \cap C' .$$
But then, since the currency is good for the rich, this can be repeated starting from $C'$ and all higher values above until we reach the top of the currency (as we can go in steps of $\alpha$, the top means $c_{\sup}- \alpha$), and so
$$ \BalanceII{V}{V}{C} \geq c_{\sup} - \val(C) - \alpha . $$
In the special case where $c_{\sup}=\infty$, note that by the same rationale as above we can reach any resource in the currency, and so in this case $\BalanceII{V}{V}{C} = c_{\sup} - \val(C) = \infty $.

Finally, note that the balance starting from $C$ is trivially upper bounded by $c_{\sup} - \val(C)$ or else one could exceed the currency limit by performing the operation, that is, trivially
$$ c_{\sup} - \val(C) \geq \BalanceII{V}{V}{C}.$$
\end{proof}

\begin{proposition}[Cost reaches the top]
\label{prop:cost_top}
In a resource theory with a currency $C\in \mathcal C$ that is good for the rich, 
a resource $V\in\mathcal S$ such that $\BalanceII{V}{V}{C}>0$ for some $C\in\mathcal C$ satisfies
$$ \Cost{V} \geq c_{\sup} - \val C - \beta $$
for some $\beta\geq \alpha + \eps + \eps' $ with $\alpha$ as in Proposition~\ref{prop:pathological_I} and $\eps,\eps'\geq 0$ arbitrarily small or zero (depending on whether or not $\Cost{V}$ and $c_{\sup}$ are achievable).
\end{proposition}
\begin{proof}
Let $C_c$ be a resource such that 
$$ C_c \to V $$
and 
$ \val(C_c) = \Cost{V}+\eps $
for some arbitrarily small or zero $\eps\geq 0$, depending on whether or not the cost of $V$ is achievable. Then, because the theory is fair for the rich, if
$ \val(C_c) + \val C < c_{\sup} $
then there is a resource $C'\in\mathcal C$ such that
$$ C' \to C \cap V  $$
and 
$ \val (C') = \val(C_c) + \val(C) $.
But then we can use the fact that
$$ \BalanceII{V}{V}{C} \geq c_{\sup} - \val(C) - \alpha $$
to see that (unless $c_{\sup}=\infty$) then in fact
$$ C' \to V \cap C \to V \cap C'' $$
for some $C''\in\mathcal C$ satisfying
$$ \val(C'') \geq c_{\sup} - \alpha - \eps' $$
for arbitrarily small or zero $\eps'\geq 0$ (depending on whether or not $c_{\sup}$ is achievable by a currency resource). But then since in particular
$ C' \to C'' $, we must have that
$ \val(C') \geq \val C''$ and so
$ \val(C') \geq c_{\sup} - \alpha - \eps' $.
But then
$$ \Cost{V} + \eps + \val(C) = \val(C_c) + \val(C) = \val(C') \geq c_{\sup} - \alpha - \eps' $$
and so
$$ \Cost{V} \geq c_{\sup} - \val(C) - \alpha - \eps - \eps' .$$
Finally, if $c_{\sup}=\infty$ then it has to be the case that also $\Cost{V}=\infty$, because otherwise we could follow the rationale above, and starting from $C_c$ generate arbitrarily large currency values larger than $\val(C_c)$, which together with the monotonicity of value yields a contradiction.
\end{proof}

\begin{proposition}[Resources with smaller cost than yield]
\label{prop:condition_balance_positive}
In a theory with a currency $\mathcal C$ that is good for the rich, for any resource $V\in\mathcal S$
\begin{align*} 
&\Cost{V}<\Yield{V} \\
&\implies \BalanceII{V}{V}{C}> 0
\end{align*}
for all $C\in\mathcal C$. Furthermore, if $\Cost{V}<\Yield{V}$, then also
$$ \Yield{V} \geq c_{\sup} - \alpha  $$
for $\alpha$ equal to the value of the smallest currency above $0$ \footnote{if there is no one smallest, we may choose one arbitrarily close to zero}.
\end{proposition}
\begin{proof}
Let $V$ be a resource such that $\Cost{V}<\Yield{V}$. Then there are currency resources $C_c,C_y\in\mathcal C$ with achievable cost and yield values such that
$$ \Cost{V}\leq\val C_c < \val C_y \leq \Yield{V} $$ 
and
$$ C_c \to V, \ V \to C_y $$ 
are possible processes. 
But then, because anything that can be done with $C_c$ can still be done with $C_y$, in particular, one can buy $V$ and have some currency resource left over. This means that there is also a currency resource 
$ C_\eps $ with 
$\val{C_\eps}=\eps=\val C_y - \val C_c > 0 $
such that
$$ V \to C_y \to V \cap C_\eps $$
is a possible process. \\
Then clearly 
$ \BalanceII{V}{V}{\Omega}\geq \eps > 0 $
and so because the currency is fair for the rich,
$$ \BalanceII{V}{V}{C}\geq \eps > 0 $$
for all $C\in\mathcal C$.

Furthermore, due to Proposition~\ref{prop:pathological_I} and $\val(\Omega)=0$, then in fact
$$ \BalanceII{V}{V}{\Omega}\geq c_{\sup} - \alpha $$
for $\alpha$ equal to the value of the smallest currency value above $0$ (if there is no one smallest, it could be any arbitrarily small resource). But then
$$ \Yield{V}\geq \BalanceII{V}{V}{\Omega} \geq c_{\sup} - \alpha .$$

\end{proof}

\subsubsection{More does not help, fairness for the poor}
\begin{proposition}[Pathological case II]
\label{prop:pathological_II}
Let $\mathcal C$ be a currency that is fair for the poor. Then for pathological resources $V$ such that
$\BalanceII{V}{V}{C}>0$ for some $C\in\mathcal C$,
$$ \Yield{V}\geq\BalanceII{V}{V}{C}>0=\Cost{V}. $$
\end{proposition}
\begin{proof}
We will first show that if the currency is good for the poor and $\BalanceII{V}{V}{C}>0$ for some $C\in\mathcal C$, then $\Cost{V}=0$.

Namely, let $C_c$ be a currency resource such that $C_c \to V$. But then, if
$\BalanceII{V}{V}{C} > 0$ for some $C\in\mathcal C$, then also
$\BalanceII{V}{V}{\Omega} > 0 $
due to fairness for the poor. 
This means that there is a $C_\eps$ with $\val C_\eps = \eps > 0$ such that we find
$$ C_c \to V \to V \cap C_\eps .$$
Hence, because the theory is good for the poor in fact there is a resource $C'_c\in\mathcal C$
with $\val C'_c = \val C_c - \eps $
such that 
$$ C'_c \to V . $$
We can repeat the argument many times until we 
reach $\val(C'_c) < \eps$. In this case $\BalanceII{\Omega}{V}{C'_c}>0$ and so clearly also 
$\Cost{V}=0$.

To show that $\Yield{V}\geq\BalanceII{V}{V}{C}$, we again use that  
$\BalanceII{V}{V}{\Omega}\geq \BalanceII{V}{V}{C}>0$, and 
also that $ V \to V \cap C'$ is more difficult to perform than $ V \to  C' $ for any $V$ and $C'$ | so trivially
$\Yield{V}\geq\BalanceII{V}{V}{C}$.

\end{proof}

\subsubsection{Both directions}

\begin{corollary}[Pathological resources in a Midas theory]
\label{cor:Midas_I}
In a theory with a fair currency, for any resource $V$ in the target
$$ \Cost{V} < \Yield{V} \iff \Balance{V}{V}>0 .$$
\end{corollary}
\begin{proof}
This follows directly from Propositions~\ref{prop:pathological_II} and~\ref{prop:condition_balance_positive}.

\end{proof}

\begin{corollary}
\label{cor:Midas_II}
In a theory with a fair currency and a Midas resource $V$, it holds that
$$ \Yield{V} \geq c_{\sup} - \alpha $$
and
$$ \Balance{V}{V} \geq c_{\sup} - \alpha $$
for some margin $\alpha$ that can be chosen as $\alpha=\val C_\alpha$ for any currency resource $C_\alpha$ with very small value (but above zero). \\
Furthermore, such a theory is essentially trivial in the sense that also
$$ \Balance{\Omega}{\Omega} \geq c_{\sup} - \alpha .$$
\end{corollary}
\begin{proof}
The first two statements follow immediately from Proposition~\ref{prop:condition_balance_positive} and Proposition~\ref{prop:pathological_I} by choosing $C=\Omega$. For the last statement, note that since we can find $C_c,C_y\in\mathcal C$ such that
$$ C_c \to V \to C_y $$ with 
$\val C_c < \val C_y $, 
$$ \BalanceII{\Omega}{\Omega}{C_c}>0 .$$
But then due to fairness this is true for the unique balance,
$$ \Balance{\Omega}{\Omega}>0 .$$
\end{proof}

\section{A short note on monotones}
\label{appendix:monotones}

The question of whether or not a particular resource can be transformed into another   can in general be difficult to answer. This may already be the case for traditional resource theories, where the task corresponds to identifying the pre-order on state space, but is even more involved for specifications as one may have to deal with large sets.

Just like with currencies, often we can decide whether or not $V \to W$ more easily by means of \emph{monotones}. These are functions on resources that are easy to calculate (otherwise they are not that useful) and monotonous under $\to$. Then, if a monotonously decreasing function $f$ (such as the free energy for a thermodynamic process), if $f(V)\geq f(W)$ does not hold, we immediately know that $V\nrightarrow W$. 

\begin{definition}[Monotones]
Let $(S^\Omega, \cT)$ be a resource theory. 
A \emph{monotone} is a function $M: S^\Omega \to \mathbb R$ such that
$$V \to W \implies M(V) \geq M(W).$$
A \emph{complete family of monotones} is a family of real functions $\{M_i\}_{i \in \I }$ such that
\begin{align*}
V \to W \Leftrightarrow M_i(V) &\geq M_i(W).
\end{align*}
Note: we could equally have $\leq$ instead of $\geq$ in the expressions above.
\end{definition}

A complete family of monotones is hence a set of monotones that completely characterizes the pre-order $\to$. This, however, can in principle be an infinitely large set, and so it is not always a useful means of deciding whether $V\to W$ in the resource theory. Examples of a complete family of monotones would be the R\'enyi entropies for the trumping relation $\prec_T$ or the $\alpha$-free energies for thermodynamic processes with exact catalysis~\cite{Brandao2013b}.

Monotones are related to currencies: if a theory has a currency, $\Cost{\cdot}$ and $\Yield{\cdot}$ are always monotones (Theorem~\ref{thm:currency_conditions}). 
However, in the other direction not any monotone would constitute a good characterization of value. 
This is 
because the actual numbers assigned through the monotone need not have any operational significance beyond their monotonicity. For example, intuitively we would expect that the difference in value of two resources should be linked to a \emph{cost} of transforming one into the other: supplying another resource with the missing value should enable an agent to implement the transformation. We would thus expect a value function to satisfy certain (sub- or super-) additivity properties, depending on its exact operational meaning. This is of course not guaranteed for general monotones. 

On the other hand, while currencies ensure an operational meaning of the monotones $\Cost{\cdot}$ and $\Yield{\cdot}$, the cost and yield of resources in general do  not constitute a complete set of monotones, as the theory may have a richer  structure. If the currency is tight, however, the monotonicity of the cost of resources actually constitutes a sufficient condition for resource conversion (Theorem~\ref{thm:tightness_order}).


\newpage

\renewcommand{\lofname}{List of theorems and propositions}

\listoftheorems[ignoreall,
show={theorem,proposition}]

\renewcommand{\lofname}{List of definitions}

\listoftheorems[ignoreall,
show={definition}]



\end{document}